\documentclass{article}

\usepackage[standard]{ntheorem}

\usepackage{microtype}
\usepackage{amssymb}
\usepackage{hyperref}
\usepackage{calrsfs}
\DeclareMathAlphabet{\pazocal}{OMS}{zplm}{m}{n}
\usepackage{amsmath}
\usepackage{xspace}
\usepackage{colonequals}
\usepackage{hyperref}
\usepackage{subcaption}
\usepackage{stmaryrd}
\usepackage{enumitem}
\usepackage{hyperref}
\usepackage[normalem]{ulem}
\usepackage{comment}
\usepackage{tikz}
\usetikzlibrary{arrows.meta,shapes,snakes,automata,backgrounds,petri,positioning,er,arrows,decorations.text,}
\usepackage{tikz-qtree}
\usepackage{tikz-qtree-compat}
\usepackage[linesnumbered,algoruled,boxed,lined]{algorithm2e}

\newenvironment{sketch}{}{}
\newenvironment{fullproof}{}{}

\includecomment{sketch}
\excludecomment{fullproof}

\def\qed {{                
   \parfillskip=0pt        
   \widowpenalty=10000     
   \displaywidowpenalty=10000  
   \finalhyphendemerits=0  
   \leavevmode             
   \unskip                 
   \nobreak                
   \hfil                   
   \penalty50              
   \hskip.2em              
   \null                   
   \hfill                  
   $\square$
   \par}}                  

\newcommand{\mycomment}[1]{{\color{gray}\par\smallskip\fbox{\parbox{\linewidth}{#1}}\par\smallskip}}

\newcommand{\interval}[1]{{\ensuremath{\mathord{\text{\normalfont\fontfamily{lmtt}\selectfont{}#1}}}}}
\newcommand{\NONE}{\interval{0}\xspace}
\newcommand{\ONE}{\interval{1}\xspace}
\newcommand{\MAYBE}{\interval{?}\xspace}
\newcommand{\MANY}{\interval{*}\xspace}
\newcommand{\PLUS}{\interval{+}\xspace}

\newcommand{\Rbold}{\mathbf{R}}
\newcommand{\Rcal}{\mathcal{R}}
\newcommand{\fd}{\mathrm{fd}}

\newcommand{\Iri}{\mathsf{Iri}}
\newcommand{\NullIri}{\mathsf{NullIri}}
\newcommand{\Lit}{\mathsf{Lit}}
\newcommand{\NullLit}{\mathsf{NullLit}}
\newcommand{\Pred}{\mathsf{Pred}}
\newcommand{\nodes}{\mathit{nodes}}

\newcommand{\Sbold}{\mathbf{S}}
\newcommand{\Tcal}{\mathcal{T}}
\newcommand{\Literal}{\mathit{Literal}}
\newcommand{\typing}{\mathit{typing}}
\newcommand{\dbl}{\mathbin{::}}

\newcommand{\Triple}{\mathit{Triple}}

\newcommand{\TP}{\ensuremath{\textsf{\textbf{TP}}}\xspace}
\newcommand{\PF}{\ensuremath{\textsf{\textbf{PF}}}\xspace}
\newcommand{\PE}{\ensuremath{\textsf{\textbf{PE}}}\xspace}

\newcommand{\Fbold}{\mathbf{F}}
\newcommand{\Fcal}{\mathcal{F}}
\newcommand{\E}{\mathcal{E}}
\newcommand{\st}{\mathrm{st}}
\newcommand{\sol}{\mathit{sol}}

\newcommand{\tp}{\TP}
\newcommand{\pf}{\PF}

\newcommand{\sem}[1]{\llbracket #1 \rrbracket}
\newcommand{\simulates}{\twoheadrightarrow}
\newcommand{\bisimulates}{{\twoheadleftarrow\!\!\!\!\!\twoheadrightarrow}}

\newcommand{\U}{\mathcal{U}}
\newcommand{\Req}{\mathit{Req}}
\newcommand{\types}{\mathit{types}}
\newcommand{\Fbb}{\mathbb{F}}
\newcommand{\Q}{\mathcal{Q}}

\newcommand{\Acc}{\operatorname{\mathit{Acc}}}
\newcommand{\Rev}{\operatorname{\mathit{Rev}}}
\newcommand{\Eq}{\operatorname{\mathit{Eq}}}
\newcommand{\Cont}{\operatorname{\mathit{Cont}}}
\newcommand{\rr}{r\xspace}
\newcommand{\p}{p\xspace}
\newcommand{\q}{q\xspace}

\title{Consistency and Certain Answers in Relational to RDF Data Exchange with Shape Constraints}

\author{Iovka Boneva, S\l{}awek Staworko, and Jose Lozano}

\begin{document}

\maketitle
\begin{abstract}
  We investigate the data exchange from relational databases to RDF graphs
  inspired by R2RML with the addition of target shape schemas. We study the
  problems of \emph{consistency} i.e., checking that every source instance
  admits a solution, and \emph{certain query answering} i.e., finding answers
  present in every solution. We identify the class of \emph{constructive
    relational to RDF data exchange} that uses IRI constructors and full tgds (with no existential variables) in its source to
  target dependencies. We
  show that the consistency problem is coNP-complete. We introduce the notion of
  \emph{universal simulation solution} that allows to compute certain query answers to any class of
  queries that is \emph{robust under simulation}. One such class are nested
  regular expressions (NREs) that are \emph{forward} i.e., do not use the
  inverse operation. Using universal simulation solution renders tractable
  the computation of certain answers to forward NREs (data-complexity). Finally, we
  present a number of results that show that relaxing the restrictions of the
  proposed framework leads to an increase in complexity.
\end{abstract}

\section{Introduction}
\label{sec:introduction}
The recent decade has seen RDF raise to the task of interchanging data between
Web applications~\cite{michel:2013}.  In many applications the data is stored in
a relational database and only exported as RDF, as evidenced by the
proliferation of languages for mapping relational databases to RDF, such as
R2RML \cite{das:2011}, Direct Mapping \cite{dmlDB2RDF} or YARRRML~\cite{yarrrml}. As an example, consider the
following R2RML mapping, itself an RDF presented in turtle syntax
\begin{small}%
\begin{align*}
&\textsf{<\#EmpMap>}\\[-4pt]
&\quad\textsf{rr:logicalTable  [\ rr:sqlQuery\ "SELECT id, name, email FROM Emp NATURAL JOIN Email" ];}\\[-4pt]
&\quad\textsf{rr:subjectMap [ 
     rr:template "emp:\{id\}"; 
     rdf:type :TEmp
  ];}\\[-4pt]
&\quad\textsf{rr:predicateObjectMap [
     rr:predicate :name;
     rr:objectMap [ rr:column "name"]
  ];}\\[-4pt]
&\quad\textsf{rr:predicateObjectMap [
     rr:predicate :email;
     rr:objectMap [ rr:column "email"]
  ].}
\end{align*}%
\end{small}%
It exports the join of two relations
$\mathit{Emp}(\underline{\mathit{id}},\mathit{name})$ and
$\mathit{Email}(\underline{\mathit{id}},\mathit{name})$ into a set of
triples. For every employee it creates a dedicated Internationalized Resource
Identifier (IRI) consisting of the prefix \texttt{emp:} and the employee
identifier. More importantly, the class (\texttt{rdf:type}) of each employee IRI is declared
as \textsf{:TEmp}.

RDF has been originally proposed schema-less to promote its adoption but the
need for schema languages for RDF has been since identified~\cite{w3c:2013,
  gayo:17}. One of the benefits of working with data conforming to a schema is
an increased execution safety: applications need not to worry about handling
malformed or invalid data that could otherwise cause undesirable and difficult
to predict side-effects. One family of proposed schema formalisms for RDF is
based on \emph{shape constraints} and this class includes \emph{shape
  expressions schemas} (ShEx) \cite{shexspec,staworko:2015a,boneva:hal-01590350}
and \emph{shape constraint language} (SHACL) \cite{shaclspec,shacl-rec}. The two
languages allow to define a set of types that impose structural constraints on
nodes and their immediate neighborhood in an RDF graph. For instance, the type
$\textsf{:TEmp}$ has the following ShEx definition
\begin{small}%
\begin{align*}
  &\textsf{:TEmp\quad\{\ 
    :name\ xsd:string;\ :email\ xsd:string\MAYBE;\ :works\ @:TDept\PLUS\ \}}
\end{align*}%
\end{small}%
Essentially, every employee IRI must have a single $\textsf{:name}$ property, an
optional $\textsf{:email}$ property, and at least one $\textsf{:works}$
property each leading to an IRI satisfying type $\textsf{:TDept}$.

In the present paper we formalize the process of exporting a relational database
to RDF as \emph{data exchange}, and study two of its fundamental problems:
\emph{consistency} and \emph{certain query answering}. In data exchange the
mappings from the source database to the target database are modeled with
\emph{source-to-target tuple-generating dependencies} (st-tgds). For mappings
defined with R2RML we propose a class of \emph{constructive} st-tgds, which use
\emph{IRI constructors} to map entities from the relational database to IRIs in
the RDF. For instance, the R2RML mapping presented before can be expressed with
the following st-tgd
\begin{align*}
\mathit{Emp}(\mathit{id},\mathit{name})\land
\mathit{Email}(\mathit{id},\mathit{email})\Rightarrow{}
&\mathit{Triple}(\mathit{emp2iri}(id),\textsf{:name},\mathit{name})\land{}\\
&\mathit{Triple}(\mathit{emp2iri}(id),\textsf{:email},\mathit{email})\land{}\\
&\mathsf{TEmp}(\mathit{emp2iri}(id)),
\end{align*}%
where $\mathit{emp2iri}$ is an IRI constructor that generates an IRI for each
employee. The above tgd is \emph{full} i.e., it does not use existential
quantifiers. To isolate the concerns, in our analysis of the st-tgds we refrain
form inspecting the definitions of IRI constructors and require only that they
are \emph{non-overlapping}, i.e. no two IRI constructors are allowed to output
the same IRI. We focus on full constructive st-tgds used with a set of
non-overlapping IRI constructors and call this setting \emph{constructive
  relational to RDF data exchange}. We report that in this setting all 4 use
cases of R2RML~\cite{usecaseR2RML} can be expressed. Furthermore, we can cover
38 out of 54 test cases for R2RML implementations~\cite{testR2RML}: 9 test cases
use pattern-based function to transform data values and 7 test cases use SQL
statements with aggregation functions. In fact, our assessment is that the
proposed framework allows to fully address all but one out of the 11 core
functional requirements for R2RML~\cite{usecaseR2RML}, namely the \textsl{Apply
  a Function before Mapping}. Finally, in our investigations we restrict our
attention to class of \emph{deterministic shape schemas} that are at the
intersection of ShEx and SHACL, are known to have desirable computational
properties while remaining practical, and posses a sought-after feature of
having an equivalent graphical representation (in the form of shape
graphs)~\cite{Staworko:2019}.

For a given consistent source relational instance, a \emph{solution} to data
exchange is a target database (an RDF graph in our case) that satisfies the
given set of st-tgds and the target schema (a shape schema in our case). The
number of solutions may vary from none to infinitely many. The problem of
\emph{consistency} is motivated by the need for static verification tools that
aim to identify potentially erroneous data exchange settings: a data exchange
setting, consisting from the source schema, the set of st-tgds, and the target
schema is \emph{consistent} iff every consistent source database instance admits
a solution. Because many solutions may be possible, the standard \emph{possible
  word semantics}~\cite{Imielinski:1984,Abiteboul89onthe} is applied when
evaluating queries: a \emph{certain answer} to a query over the target schema is
an answer returned in every solution. Consequently, one is inclined to construct
a solution that allows to easily compute certain answers. In the case of
relational data exchange, \emph{universal solutions} have been identified and
allow to easily compute certain answers to conjunctive queries, or any class of
queries preserved under homomorphism for that
matter~\cite{fagin:2005a}. Unfortunately, for relational to RDF data exchange
with target shape schema, a finite universal solution might not exists even if
the setting is consistent and admits solutions. Also, the class of conjunctive
queries, while adequate for expressing queries for relational databases, is less
so for RDF. Query languages, like SPARQL, allowing regular path expressions with
nesting have been proposed to better suit the needs of querying
RDF~\cite{PEREZ2010255}.

The list of contributions of the present paper follows. 
\begin{itemize}
\item We formalize the framework of relational to RDF data exchange with target
  shape schema and IRI constructors, and we identify the class of
  \emph{constructive relational to RDF data exchange} that uses deterministic
  shape schemas and full constructive source-to-target dependencies.
\item We provide an effective characterization of consistency of constructive
  relational to RDF data exchange settings and show that the problem is
  coNP-complete.
\item We show that allowing nondeterministic target schemas makes the
  consistency problem $\Pi_2^p$-hard. We also present a generalization of our
  consistency characterization to include st-tgds with existential quantifiers
  but the extension is no longer in coNP and the lower bound remains an open
  question.
\item We propose a novel notion of \emph{universal simulation solution} that can
  be constructed for any consistent constructive relational to RDF data exchange
  setting. It allows to easily compute certain answers to any query class that
  is robust under graph simulation. We also apply existing results on relational
  to relational data exchange setting to show tractability of computing certain
  answers to conjunctive queries.
\item We use the universal simulation solution to show tractability of computing
  certain answers to \emph{forward nested regular expressions}. For the full
  class of \emph{nested relational expressions} (NREs), considered to be the
  navigational core of SPARQL~\cite{PEREZ2010255}, we show an increase of complexity
  when computing certain answers. 
\end{itemize}
In \cite{boneva:18} we have studied the consistency problem for a more
restrictive fully-typed data exchange setting, where all constructed IRIs must
be typed. This restriction allowed to reduce the consistency problem to a simple
test of functional dependencies propagation over relational views. This
technique can no longer be employed for constructive data exchange setting,
where the constructed RDF nodes need not be typed, and to address it we propose
a novel and non-trivial technique. In \cite{boneva:18}, we have also not
considered certain query answering.

\paragraph{Organization} The paper is organized as follows. In
Section~\ref{sec:example} we introduce the constructive data exchange framework
with an illustrative example. In Section~\ref{sec:preliminaries} we recall basic
notions of relational and graph databases. In Section~\ref{sec:data-exchange} we
formalize the relational to RDF data exchange with IRI constructors and target
shape schema. In Section~\ref{sec:consistency-des} we study the problem of
consistency. In Section~\ref{sec:query-answering} we study certain query
answering. Section~\ref{sec:related} contains a discussion of related work. And
in Section~\ref{sec:future} we present conclusions and outline future work.

\section{Introductory Example}
\label{sec:example}
\newcommand{\Bug}{\ensuremath{\mathsf{TBug}}\xspace}
\newcommand{\User}{\ensuremath{\mathsf{TUser}}\xspace}
\newcommand{\Emp}{\ensuremath{\mathsf{TEmp}}\xspace}
\newcommand{\Test}{\ensuremath{\mathsf{TTest}}\xspace}
\newcommand{\Sys}{\ensuremath{\mathsf{TSys}}\xspace}
\newcommand{\col}{\mathord{\kern-1pt\interval{:}\kern-1.25pt}}
\newcommand{\mel}{\ensuremath{\col\mathsf{email}}\xspace}
\newcommand{\name}{\ensuremath{\col\mathsf{name}}\xspace}
\newcommand{\descr}{\ensuremath{\col\mathsf{descr}}\xspace}
\newcommand{\rep}{\ensuremath{\col\mathsf{rep}}\xspace}
\newcommand{\repr}{\ensuremath{\col\mathsf{repro}}\xspace}
\newcommand{\rel}{\ensuremath{\col\mathsf{related}}\xspace}
\newcommand{\phone}{\ensuremath{\col\mathsf{phone}}\xspace}
\newcommand{\follows}{\ensuremath{\col\mathsf{follows}}\xspace}
\newcommand{\tracks}{\ensuremath{\col\mathsf{tracks}}\xspace}
\newcommand{\grp}{\ensuremath{\col\mathsf{covers}}\xspace}
\newcommand{\prep}{\ensuremath{\col\mathsf{prepare}}\xspace}
\newcommand{\access}{\ensuremath{\col\mathsf{access}}\xspace}
\newcommand{\adm}{\ensuremath{\col\mathsf{admin}}\xspace}
\newcommand{\of}{\ensuremath{\col\mathsf{of}}\xspace}
\newcommand{\RBug}{\textit{Bug}\xspace}
\newcommand{\RUser}{\textit{User}\xspace}
\newcommand{\REmail}{\textit{Email}\xspace}
\newcommand{\RRel}{\textit{Rel}\xspace}
\newcommand{\RAcc}{\textit{Access}\xspace}
\newcommand{\uToI}{\textit{usr2iri}\xspace}
\newcommand{\pToI}{\textit{pers2iri}\xspace}
\newcommand{\bToI}{\textit{bug2iri}\xspace}
\newcommand{\sToI}{\textit{sys2iri}\xspace}
\newcommand{\tToI}{\textit{test2iri}\xspace}
\newcommand{\RSys}{\textit{Sys}\xspace}
\newcommand{\RSubs}{\textit{Subscr}\xspace}
\newcommand{\RTest}{\textit{Test}\xspace}
\newcommand{\RTrack}{\textit{Track}\xspace}

We illustrate the relational to RDF framework with the following example. We work
with a \emph{relational database} of software bug reports, presented in
Figure~\ref{fig:db}. Each bug is reported by a user and a bug may have a number
of related bugs. Each user may track a number of bugs.
\begin{figure}[htb]
  \centering
  \begin{tabular}[t]{cc}
    \multicolumn{2}{c}{\RUser}\\
    \uline{\textit{uid}}& \textit{name}  \\[2pt]
    \hline
    1 & Jose \\
    2 & Edith \\
  \end{tabular}
  \hfill
  \begin{tabular}[t]{cc}
    \multicolumn{2}{c}{\REmail}\\
    \uline{\textit{uid}}& \textit{email} \\[2pt]
    \hline
    1 & j@ex.com
  \end{tabular}
  \hfill
  \begin{tabular}[t]{ccc}
    \multicolumn{2}{c}{\RTrack}\\
    \uline{\textit{uid}}& \uline{\textit{bid}}\\[2pt]
    \hline
    1 &  1\\ 
    1 &  2
  \end{tabular}
  \hfill
	\begin{tabular}[t]{ccc}
     \multicolumn{3}{c}{\RBug}\\
     \uline{\textit{bid}}& \textit{descr} & \textit{uid}\\[2pt]
     \hline
     1 & Boom! & 1 \\
     2 & Kabang! & 1 \\
     3 & Bang! & 2 
	\end{tabular}
	\hfill
	\begin{tabular}[t]{cc}
     \multicolumn{2}{c}{\RRel}\\
     \uline{\textit{bid}}& \uline{\textit{rid}} \\[2pt]
     \hline
     2 & 1 \\
     1 & 3
	\end{tabular}
   \caption{Relational source database\label{fig:db}}
\end{figure}

We wish to export the contents of the above relational database to RDF for use
by an existing application. 
The application expects the RDF document to adhere to the following ShEx schema
(with \texttt{:} being the default prefix).
\begin{align*}
  &\textsf{:\Bug\hspace{10pt}\{\ \descr\ xsd:string;\ \rep\ @:\User;\ :rel\ @\Bug\ \MANY\ \}}\\
  &\textsf{:\User\hspace{8.5pt}\{\ \name\ xsd:string;\ \mel\ xsd:string;\ \tracks\ @\Bug\ \PLUS\ \}}
\end{align*}
This schema defines two types of nodes: $\Bug$ for bug reports and $\User$ for
user info. This ShEx schema happens to closely mimic the structure of the
relational database with two exceptions: the type $\User$ requires that every
user must track at least one bug and must have a single email while the
relational database is free of such constraints.

To assign an IRI to every user and every bug, we define two IRI constructors
using the intuitive syntax of subject patterns of R2RML (where \texttt{bug:} and
\texttt{usr:} are two IRI prefixes):
\begin{align*}
&\bToI(\mathit{bid}) = \texttt{"bug:\{}\mathit{bid}\texttt{\}"}&
&\uToI(\mathit{uid}) = \texttt{"usr:\{}\mathit{bid}\texttt{\}"}
\end{align*}
Now, the R2RML mapping is formalized using the following set full constructive
dependencies.
\begin{align*}
  \RBug(b,d,u) \Rightarrow{}& \Triple(\bToI(b),\descr,d) \land \Bug(\bToI(b)) \land{}  \\
                            &\Triple(\bToI(b),\rep,\uToI(u))\\
  \RRel(b_1,b_2) \Rightarrow{}& \Triple(\bToI(b_1),\rel,\bToI(b_2))\\
  \RUser(u,n) \Rightarrow{}& \Triple(\pToI(u),\name,n)\\
  \RUser(u,n) \land \RTrack(u,b)  \Rightarrow{}&\Triple(\uToI(u),\tracks,\bToI(b))\\
  \RUser(u,n) \land \REmail (u,e) \Rightarrow{}& \Triple(\uToI(u), \mel, e)	
\end{align*}
One possible solution to the task at hand is presented in Figure~\ref{fig:rdf}.
\begin{figure}[htb]
  \centering 
  \begin{tikzpicture}[>=latex,scale=1.125,semithick]  
    
    \node (d3) at (-4,2.3) {\it ``Kabang!''};
    \node (d1) at (-1.5,2.5) {\it  ``Boom!''};
    \node (d4) at (1.15,2.5) {\it ``Bang!''};
    
    \node (bug3) at (-4,1) {$\mathsf{bug\col 2}$};
    \node (bug1) at (-1.5,1.25) {$\mathsf{bug\col 1}$};
    \node (bug4) at (1.15,1) {$\mathsf{bug\col 3}$};
    
    \node (user1) at (-2.5,0) {$\mathsf{usr}\col\mathsf{1}$};
    
    \node (emp1) at (0.9,-0.15) {$\mathsf{usr\col 2}$};  
    
    \node (n1) at (-3.5,-1.5) {\it ``Jose''};
    \node (e1) at (-1.75,-1.5) {\it ``j@ex.com''};
    \node (n3) at (0.1,-1.55) {\it ``Edith89''};
    
    \node (e3) at (1.5,-1.55) {$\bot_1$};
    \node (bug2) at (4,1.25) {$\bot_2$};
    \node (d2) at (4,2.5) {$\bot_3$};
    \node (user2) at (3.75,0) {$\bot_4$};    
    \node (e2) at (3.1,-1.5) {$\bot_5$};
    \node (n2) at (4.7,-1.5) {$\bot_6$};

    \node[right] (TBug) at (-6.25,1.125) {\texttt{:}\Bug};
    \node[right] (TUser) at (-6.25,0) {\texttt{:}\User};

    \begin{scope}[thin,green!40!black!75!white,>=stealth']
      \draw[bend angle=10] (user1) edge[->,bend left] (TUser);
      \draw[bend angle=12] (emp1) edge[->,bend left] (TUser);
      \draw[bend angle=17.5] (user2) edge[->,bend left] (TUser);
      \draw[bend angle=15] (bug3) edge[->] (TBug);
      \draw[bend angle=12] (bug1) edge[->,bend right] (TBug);
      \draw[bend angle=16] (bug4) edge[->,bend right] (TBug);
      \draw[bend angle=18] (bug2) edge[->,bend right] (TBug);
    \end{scope}
    
    \draw (user1) edge[->] node[above,sloped] {\name} (n1);
    \draw (user1) edge[->] node[above,sloped] {\mel} (e1);
    \draw (user1) edge[->,bend left=30] node[below,sloped] {\tracks} (bug3);
    \draw (user1) edge[->,bend right=50] node[below,sloped] {\tracks} (bug1);
    \draw (user2) edge[->] node[above,sloped] {\name} (n2);
    \draw (user2) edge[->] node[above,sloped] {\mel} (e2);
    \draw (user2) edge[->,bend right=70] node[below,sloped] {\tracks} (bug2);
    
    \draw (emp1) edge[->] node[above,sloped] {\name} (n3);
    \draw (emp1) edge[->] node[above,sloped] {\mel} (e3);
    \draw (emp1) edge[->,bend left=5] node[above,sloped] {\tracks} (bug2);
    \draw (bug1) edge[->] node[above,sloped] {\rel} (bug4);
    \draw (bug3) edge[->] node[above,sloped] {\rel} (bug1);
    \draw (bug1) edge[->] node[above,sloped] {\rep} (user1);
    \draw (bug1) edge[->] node[above,sloped] {\descr} (d1);
    
    \draw (bug2) edge[->] node[above,sloped] {\rep} (user2);
    \draw (bug2) edge[->] node[above,sloped] {\descr} (d2);
    
    \draw (bug3) edge[->] node[above,sloped] {\rep} (user1);
    \draw (bug3) edge[->] node[above,sloped] {\descr} (d3);
    
    \draw (bug4) edge[->] node[below,sloped] {\rep} (emp1);
    \draw (bug4) edge[->] node[above,sloped] {\descr} (d4);
  \end{tikzpicture}
  \caption{Target RDF graph (solution). Green thin arrows indicate types of
    non-literal nodes. \label{fig:rdf}}
\end{figure}
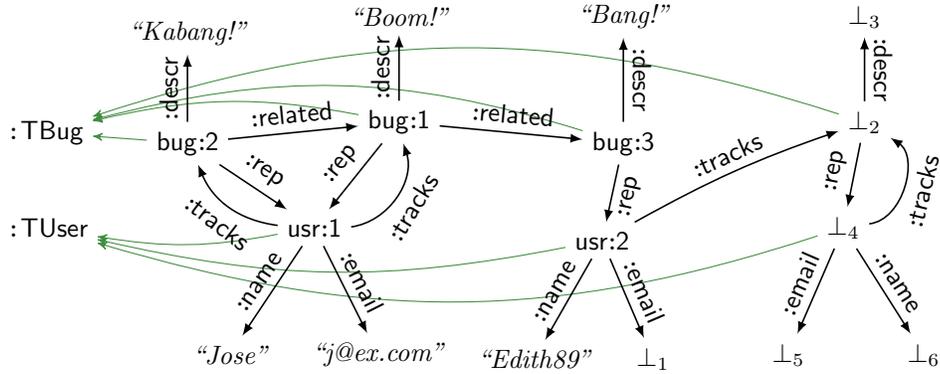
We point out that a number of null values, for both IRI and literal nodes, has
been introduced in the solution to make sure it satisfies the shape schema.

\section{Preliminaries}
\label{sec:preliminaries}
In this section we recall basic notions of relational and graph databases. More
formal definitions can be found in appendix.

\paragraph{Relational databases} 
A \emph{relational schema} is a pair $\Rbold = (\Rcal, \Sigma_\fd)$ where
$\Rcal$ is a set of relation names and $\Sigma_\fd$ is a set of functional
dependencies. Each relation name has a fixed arity and a set of attribute
names. A \emph{functional dependency} is written as usual $R:X\rightarrow Y$
where $R$ is a relation name and $X$ and $Y$ are two sets of attributes of
$R$. An \emph{instance} $I$ of $\Rbold$ is a function that maps every relation
name of $\Rbold$ to a set of tuples over a set $\Lit$ of constants (also
called literal values). $I$ is \emph{consistent} if it satisfies
all functional dependencies $\Sigma_\fd$.

\paragraph{Graphs} 
An RDF graph $G$ is an labeled graph whose nodes are divided into two
\emph{kinds}: \emph{literal} nodes and \emph{non-literal} nodes with only
non-literal nodes allowed to have outgoing edges. Every node is labeled but the
label might be a \emph{named null}. The type of value used depends on the kind
of a node: literal nodes are labeled with \emph{literal values} $\Lit$ and
\emph{literal null values} $\NullLit$ while non-literal node can be labeled with
resource names $\Iri$ and null resource names $\NullIri$. More importantly, we
adopt the \emph{unique name assumption} (UNA) i.e., no two node have the same
label, and consequently, we equate nodes with their labels and by $\nodes(G)$ we
denote the set of labels of nodes of $G$. Also, each edge is labeled with a
predicate name, which is a non-null resource name $\Pred\subset\Iri$. We often
view a graph as a set of \emph{subject-predicate-object triples}.

\paragraph{Shape Schemas}
A \emph{shapes schema} is a pair $\Sbold=(\Tcal,\delta)$, where $\Tcal$ is a
finite set of \emph{type names} and
$\delta\subseteq\Tcal\times\Pred\times(\Tcal\cup\{\Literal\})\times
\{\ONE,\MAYBE,\MANY,\PLUS\}$ is a set of shape constraints. A \emph{shape
  constraint} $(T,p,S,\mu)$ reads as follows: if a node has type $T$, then every
neighbor reached with an outgoing $p$-edge must have type $S$ and the number of
such neighbors must be within the bounds of $\mu$: precisely one if $\mu=\ONE$,
at most one if $\mu=\MAYBE$, at least one if $\mu=\PLUS$, and arbitrarily many
if $\mu=\MANY$. Naturally, the validity of a graph $G$ w.r.t.\ $\Sbold$ is
defined relative to a typing, a function
$\typing:\nodes(G)\rightarrow\Tcal\cup\{\Literal\}$ that assigns to every
non-literal node a set of types in $\Tcal$ and to every literal node the special
type label $\Literal$. A \emph{typed graph} $(G,\typing)$ is \emph{valid}
w.r.t.\ $\Sbold$ if every shape constraint of $\Sbold$ is satisfied relative to
$\typing$.

We work only with \emph{deterministic} shape schemas such that for every type
$T\in\Tcal$ and every predicate $p\in\Pred$ there is at most one shape
constraints with $T$ and $p$. Consequently, we view $\delta$ as a partial
function
$\delta:\Tcal\times\Pred\rightarrow(\Tcal\cup\{\Literal\})
\times\{\ONE,\MAYBE,\MANY,\PLUS\}$ and set $\delta(T,p)=S^\mu$ whenever
$(T,p,S,\mu)\in\delta$. We point out that deterministic shape schemas are
expressible in both ShEx and SHACL.

\paragraph{Dependencies}
We employ the standard syntax of first-order logic and given a relational schema
$\Rbold$ and a shape schema $\Sbold$, the vocabulary used to construct formulas
comprises of the relation names of $\Rbold$, a ternary predicate $\Triple$ for
defining graph topology, and the types of $\Sbold$ used as monadic
predicates. We also the edge labels $\Pred$ as constant symbols with their
straightforward interpretation. Naturally, we use of the equality relation $=$
and but by \emph{clause} we understand a conjunction of (positive) atomic
formulas that does not use $=$. Later on, we additionally introduce functions
that allow to map the values in relational databases to resource names used in
RDF graphs, and we shall allow the use of their names in formulas but without
nesting.

Now, a \emph{dependency} is a formula of the form
$\forall\bar{x}.\varphi\Rightarrow\exists\bar{y}.\psi$, where $\varphi$ is
called the \emph{body} and $\psi$ the \emph{head} of the dependency, and we
typically omit the universally quantified variables and write simply
$\varphi\Rightarrow\bar{y}.\psi$. A dependency is \emph{equality-generating}
(egd) if its body is a clause and its head consists of an equality condition
$x=y$ on pairs of variables. A \emph{tuple-generating dependency} (tgd) uses
clauses in both its head and its body. A tgd is \emph{full} if it has no
existentially quantified variables.

A number of previously introduced concepts can be expressed with
dependencies. Any functional dependency is in fact an equality-generating
dependency. For instance, the key dependency
$\RUser:\mathit{uid}\rightarrow\mathit{name}$ in the example in
Section~\ref{sec:example} can be expressed as
$\RUser(x,y_1) \land \RUser(x,y_2) \Rightarrow y_1=y_2$. Interestingly, any
deterministic shape schema $\Sbold$ can be expressed with a set $\Sigma_\Sbold$
of equality- and tuple-generating dependencies. More precisely, whenever
$\delta(T,p)=S^\mu$ the set $\Sigma_\Sbold$ contains:
\begin{description}
\itemsep0pt
\item[(\TP)] the \emph{type propagation} rule: 
  $T(x)\land\Triple(x,p,y)\Rightarrow S(y)$,
\item[(\PF)] the \emph{predicate functionality} rule:
  $T(x)\land\Triple(x,p,y_1)\land\Triple(x,p,y_2)\Rightarrow y_1=y_2$\\
  if $\mu=\ONE$ or $\mu=\MAYBE$, 
\item[(\PE)] the \emph{predicate existence} rule:
  $T(x)\Rightarrow\exists y.\ \Triple(x,p,y)$ if $\mu=\ONE$ or $\mu=\PLUS$.
\end{description} 
\paragraph{Chase}
We use the standard notion of \emph{homomorphism} and its extensions to formulas
and sets of facts (relational structures). The \emph{chase} is a procedure used
to construct a solution for data exchange, and it begins with the source
instance and iteratively executes any dependencies that are triggered. More
precisely, a dependency $\sigma = \varphi \Rightarrow \exists \bar{y}. \psi$ is
\emph{triggered} in instance $I$ by a homomorphism $h$ if
$h(\varphi)\subseteq I$ and there is no extension $h'$ of $h$ with
$h'(\psi)\subseteq I$. The \emph{execution} of $\sigma$ triggered in $I$ by $h$
may result in 1) adding new facts to $I$ when $\sigma$ is a tgd, 2) in renaming
named null in $I$ when $\sigma$ is an egd, or 3) in a \emph{failure} if $\sigma$
is an egd and $\psi$ contains a value equality $x = y$ but $h(x)$ and $h(y)$ are
two different constants.

\section{Constructive Relational to RDF Data Exchange}
\label{sec:data-exchange}
An $n$-ary \emph{IRI constructor} is a function $f:\Lit^n\to\Iri$ that maps an
$n$-tuple of database constants to an RDF resource name. A \emph{IRI
  constructor} library is a pair $\Fbold=(\Fcal,F)$, where $\Fcal$ is a set of
IRI constructor names and $F$ is their interpretation. $\Fbold$ is
\emph{non-overlapping} if all its IRI constructors have pairwise disjoint
ranges.

\begin{definition}
  A \emph{relational to RDF data exchange setting with fixed IRI constructors}
  is a tuple $\E = (\Rbold, \Sbold, \Sigma_\st,\Fbold)$, where
  $\Rbold= (\Rcal,\Sigma_\fd)$ is a source relational schema,
  $\Sbold = (\Tcal, \delta)$ is a target shape constraint schema,
  $\Fbold=(\Fcal,F)$ is an IRI constructor library, and $\Sigma_\st$ is a set of
  \emph{source-to-target tuple generating dependencies} \emph{(}st-tgds\emph{)}
  whose bodies are formulas over $\Rcal$ and heads are formulas over
  $\Fcal\cup\Tcal\cup\{\Literal\}$. $\E$ is \emph{constructive} if the library
  of IRI constructors is non-overlapping and the st-tgds $\Sigma_\st$ are full
  tgds.

  A typed graph $J$ is a \emph{solution} to $\E$ for a source instance $I$ of
  $\Rbold$, iff $J$ satisfies $\Sbold$ and $I\cup J\cup F \models\Sigma_\st$. By
  $\sol_\E(I)$ we denote the set of all solutions for $I$ to $\E$.  \qed
\end{definition}
In the reminder we fix a constructive data exchange setting $\E$, and in
particular, we assume a fixed library of IRI constructors $\Fbold$. Since we
work only with constructive data exchange settings, w.l.o.g. we can assume that
the heads of all st-tgds consist of one atom only. We point out that while a
constructive data exchange setting does not use egds, our constructions need to
accommodate egds and tgds coming from the shapes schema.

The \emph{core pre-solution} for $I$ to $\E$ is the result $J_0$ of chase on $I$
with the st-tgds $\Sigma_\st$ and all \TP rules of $\Sbold$. In essence $J_0$
isobtained by exporting the relational data to RDF triples with $\Sigma_\st$ and
then propagating any missing types according to $\Sbold$ but without creating
any new nodes with \PE rules. This process does not introduce
any null values and always terminates yielding a unique result. Naturally, $J_0$
is included in any solution $J\in\sol_\E(I)$.


\section{Consistency}
\label{sec:consistency-des}
{
  \newcommand{\Instpisigma}{I_{\pi,\sigma,\sigma'}}
  \newcommand{\Bpisigma}{B_{\pi,\sigma,\sigma'}}
  \newcommand{\hpisigma}{h_{\pi,\sigma,\sigma'}}
  \newcommand{\Elong}{\E = (\Rbold, \Sbold, \Sigma_\st, \Fbold)}
  \newcommand{\vars}{\mathit{vars}}
  \newcommand{\Cotypes}{\operatorname{\mathit{CoTypes}}}

In this section we study the problem of consistency of data exchange settings.
The following notion of consistency was called absolute consistency in \cite{murlak-xml-absolute-consistency}.

\begin{definition}[Consistency]
  A data exchange setting $\E$ is \emph{consistent} if every consistent source instance $I$ of $\Rbold$ admits a solution to $\E$.
\end{definition}

First we show that a constructive data exchange setting $\E$ is consistent if and only if it is value consistent (see Section~\ref{sec:cone}) and node kind consistent (see Section~\ref{sec:ctwo}) and the decision procedure is co-NP complete.
Then in Section~\ref{sec:consistency-non-constructive} we show that consistency checking is more complex for two more general data exchange settings.

\subsection{Value Consistency}
\label{sec:cone}
Value inconsistency captures situations in which all chase sequences would fail due to triggering a predicate functionality egd $T(x) \land \Triple(x, p, y) \land \Triple(x, p, y') \Rightarrow y = y'$ with a homomorphism that associates different constants with $y$ and $y'$.
Let $\Sigma_\Sbold^\tp$ be the set of type propagation rules and $\Sigma_\Sbold^\pf$ be the set of predicate functionality rules from $\Sigma_\Sbold$ as defined in Section~\ref{sec:data-exchange}.


\begin{definition}[Value consistent]
  \label{def:cone-instance}
  Let $J$ be the core pre-solution for some source instance $I$ to $\E$.
  $J$ is \emph{value consistent} if $J \models \Sigma_\Sbold^\pf$.
  The data exchange setting $\E$ is \emph{value consistent} if for every $I$ instance of $\Rbold$, the core pre-solution for $I$ to $\E$ is value consistent.
\end{definition}

We now concentrate on identifying whether core pre-solutions to $\E$ satisfy $\Sigma_\Sbold^\pf$.
A triple of facts $W = \{T(f(\bar{a})), \Triple(f(\bar{a}), p, b), \Triple(f(\bar{a}), p, b')\}$ is called a \emph{violation} if the definition of type $T$ contains a triple constraint of the form $p\dbl{}S^\ONE$ or $p\dbl{}S^\MAYBE$, and $b \neq b'$ are constants.
The triple $(T,f,p)$ is called the \emph{sort} of the violation.

We fix a violation $W = \{T(f(\bar{a})), \Triple(f(\bar{a}), p, b), \Triple(f(\bar{a}), p, b')\}$ for the sequel, and we explain how to check whether the dependencies in $\E$ allow to generate this violation.
%
%
The proof goes by constructing a finite set $V$ of source instances s.t. $\E$ is value inconsistent iff there is an instance $I$ in $V$ s.t. chasing $I$ with $\Sigma_\fd$ fails. 
We start by an example illustrating some elements of the decision procedure.
\begin{example}
  \label{ex:illustration}
  Let $\Elong$ where $\Rbold = (\Rcal, \Sigma_\fd)$, $\Fbold=(\Fcal,F)$, $\Rcal = \{R,S\}$ both of arity two, and $\Fcal = \{g_0, g,f\}$ all of arity one.
  The shapes schema is given by $\delta(U_0, \rr) = U^\MANY$, $\delta(U, \q) = T^\MANY$, $\delta(T,\p) =\Literal^\ONE$, and the st-tgds are as follows:
  $$
  \begin{array}{lllrl}
    \text{\scriptsize (1)}\;& R(x_0,x_1) \Rightarrow U_0(g_0(x_1)) &
                             \text{\scriptsize (4)}\;& S(x,y) &\Rightarrow \Triple(f(x), \p, y)\\
    \text{\scriptsize (2)}\;& R(x_1,x_2) \Rightarrow \Triple(g_0(x_1), \rr, g(x_2))~~~~~&
                             \text{\scriptsize (5)}\;& R(x,z) \land S(x,y') &\Rightarrow \Triple(f(x'), \p, y')\\
    \text{\scriptsize (3)}\;& R(x_2,x) \Rightarrow \Triple(g(x_2), \q, f(x))
  \end{array}
  $$
  We want to construct a source instance s.t. when chased with $\E$ would produce a violation of sort $(T,p,f)$.
  First we need to produce a fact $T(f(x))$ for some $x$.
  This can be done by applying rules {\scriptsize(1)--(3)}, then the type propagation rules for $\delta(U_0,\rr) = U^\MANY$ and $\delta(U, \q) = T^\MANY$.
  More precisely, let $I_{123}$ be the instance obtained as the union of the bodies of rules {\scriptsize(1)--(3)} (where variables are used as elements of the domain).
  Note that the variables repeated between rules were chosen in such a way on purpose.
  The result of chasing $I_{123}$ by the above mentioned rules is $I' = I_{123} \cup \{U_0(g_0(x_1)), \Triple(g_0(x_1), \rr, g(x_2)), U(g(x_2)), \Triple(g(x_2), \q, f(x)), T(f(x)) \}$.
  Now we want to use rules {\scriptsize (4),(5)} to obtain the two missing facts for the violation.
  For that, let $I_{123,4,5}$ be the union of $I_{123}$ and the bodies of rules {\scriptsize (4),(5)}.
  Chasing $I_{123,4,5}$ with $\Sigma_\st \cup \Sigma_\Sbold^\tp$ we get its core pre-solution to $\E$: $J = I' \cup \{\Triple(f(x), \p, y),  \Triple(f(x'), \p, y'), \Literal(y), \Literal(y')\}$ that contains a violation of sort $(T,p,f)$.

  So far we didn't give the source dependencies on purpose.
  Suppose that the first attribute of $S$ is a primary key.
  In this case, $I_{123,4,5}$ is not a consistent source instance, and we can actually show that $\E$ is consistent.
  Without source dependencies, $\E$ is inconsistent, as witnessed by the source instance $I_{123,4,5}$.
  \qed
\end{example}

Now we identify a necessary and sufficient condition for whether fact $T(f(\bar{a}))$ can appear in the pre-solutions to $\E$. 

\begin{definition}
\label{def:accessible}
The pair $(T,f) \in \Tcal \times \Fcal$ is called \emph{accessible in $\E$ with sequence $\sigma_0,\sigma_1, \ldots, \sigma_n$} of st-tgds in $\Sigma_\st$ if:
\begin{itemize}
\item the head of $\sigma_0$ is of the form $T_0(f_0(\bar{y}_0))$, and
\item the head of $\sigma_i$ is of the form $Triple(f_{i-1}(\bar{x}_i), p_i, f_i(\bar{y}_i))$ for every $1 \le i \le n$, and
\item $\delta(T_{i-1},p_i) = T_i^{\mu_i}$ for every $0 \le i < n$, and
\item $T = T_n$ and $f = f_n$.
\end{itemize}
for some type symbols $T_i$, function symbols $f_i$, predicates $p_i$ and sequences of variables $\bar{x}_i$ and $\bar{y}_i$.
\end{definition}
Note that if $(T,f)$ is accessible in $\E$, then it is accessible with an elementary sequence $\sigma_0, \ldots, \sigma_n$ which elements are pairwise distinct.

In Example~\ref{ex:illustration}, $(T,f)$ is accessible in $\E$ with sequence {\scriptsize (1)(2)(3)}.

The pairs $(T',f')$ accessible in $\E$ characterize the type facts that appear in the core pre-solutions to $\E$, as follows.
\begin{lemma}
  \label{lem:existence-type-fact}
  For any $(T,f) \in \Tcal \times \Fcal$ it holds:
  $(T,f)$ is accessible in $\E$
  if and only if
  there exists an instance $I$ of $\Rcal$ and a tuple of constants $\bar{a}$ in the domain of $I$ s.t. the core pre-solution for $I$ to $\E$ contains the fact $T(f(\bar{a}))$.
\end{lemma}

Now we assume that the fact $T(f(\bar{a}))$ appears in the core pre-solutions for $I$ to $\E$ and want to verify whether the facts $\Triple(f(\bar{a}), p, b), \Triple(f(\bar{a}), p, b')$ co-occur with it.
Recall that $b,b'$ are constants, so such facts are necessarily generated by st-tgds.
Two st-tgds $\sigma,\sigma'$ are called \emph{contentious} with sort $(T,p,f)$ if the head of $\sigma$ is $\Triple(f(\bar{z}), p, t)$, the head of $\sigma'$ is $\Triple(f(\bar{z}'), p, t')$ and $(T,f)$ is accessible in $\E$, and predicate p is functional for type $T$, i.e. $\delta(T,p) = S^\mu$ with $\mu$ equal to $\ONE$ or $\MAYBE$.
Note that $\sigma,\sigma'$ may be the same st-tgd, in which case we consider that they are two copies of it obtained by alpha renaming.

Suppose now that $\sigma,\sigma'$ are the contentious st-tgds here above, and that $\pi = \sigma_0, \ldots, \sigma_n$ is a sequence of st-tgds s.t. $(T,f)$ is accessible in $\E$ with $\pi$.
We define a source instance $\Instpisigma$ such that a chase sequence with rules $\sigma_0,\ldots,\sigma_n,\sigma,\sigma'$ can be executed on $\Instpisigma$ yielding an instance that includes the violation $W$.
%
Let $\sigma_{n+1} = \sigma$ and $\sigma_{n+2} = \sigma'$.
Suppose w.l.o.g. that $\sigma_i$ and $\sigma_j$ use mutually disjoint sets of variables whenever $i\neq j$.
Define $\Bpisigma = \bigcup^{n+2}_{i=0} \mathit{body}(\sigma_i)$ where $\mathit{body}(\sigma_i)$ is the body of $\sigma_i$.
Let $\sigma_0,\ldots,\sigma_n$ be as in Definition~\ref{def:accessible}, thus $(T_n,f_n,p_n) = (T,f,p)$.
%
%
Define the sequence of mappings $h_0, \ldots, h_{n+2}$ inductively as follows:
\begin{itemize}
\item for any $0 \le i \le n+2$, $h_i: \bigcup^i_{j=0} \vars(\sigma_j) \to \NullLit$ is a mapping that is injective when restricted on $\vars(\sigma_i)$, where $\vars(\sigma)$ denotes the set of variables that appear in $\sigma$.
\item for any $1 \le i \le n+2$, $h_i$ coincides with $h_{i-1}$ on the domain of $h_{i-1}$;
\item for any $1 \le i \le n$, $h_i(\bar{x}_i) = h_{i-1}(\bar{y}_{i-1})$ and $h_i(z)$ is fresh w.r.t. the image of $h_{n-1}$ for any $z \not\in \bar{x}_i$. That is, $z \not\in \bar{x}_i$ implies $h(z)$ is not in the image of $h_{i-1}$;
\item $h_{n+1}(\bar{z}) = h_n(\bar{y}_n)$ and $h_{n+1}(z)$ is fresh w.r.t. the image of $h_n$ for any $z \not\in \bar{z}$;
\item $h_{n+2}(\bar{z}') = h_n(\bar{y}_n)$ and $h_{n+2}(z)$ is fresh w.r.t. the image of $h_{n+1}$ for any $z \not\in \bar{z}'$,
\end{itemize}
Then we let $\hpisigma = h_{n+2}$ and $\Instpisigma = \hpisigma(\Bpisigma)$.
It immediately follows from the definition that $\hpisigma: \Bpisigma \to \Instpisigma$ is a homomorphism.
Moreover, it is easy to see that $\Instpisigma$ is unique up to isomorphism, so from now on by $\Instpisigma$ we mean an arbitrary instance isomorphic to the one defined above.
In Example~\ref{ex:illustration}, the instance $I_{123,4,5}$ was obtained as described above.

The following proposition establishes an equivalence between the presence of the violation $W$ in the core pre-solution of an instance $I$ of $\Rcal$, and the existence of a homomorphism from some $\Instpisigma$ to $I$.

\begin{proposition}
  \label{prop:existence-violation}
  Let $I$ be an instance of $\Rcal$.
  \begin{enumerate}
  \item \label{item:prop-existence-violation:1} There exist $\pi, \sigma, \sigma, h$ s.t. $(T,f)$ is accessible in $\E$ with path $\pi$, $\sigma, \sigma'$ are contentious st-tgds of sort $(T,f,p)$, and $h: \Instpisigma \to I$ is a homomorphism
    if and only if
    there exist a tuple of constants $\bar{a}$ from the domain of $I$ and constants $b, b'$ s.t. the core pre-solution for $I$ to $\E$ includes $\{T(f(\bar{a})), \Triple(f(\bar{a}),p,b), \Triple(f(\bar{a}),p,b')\}$.
  \item \label{item:prop-existence-violation:2} Moreover, if the head of $\sigma$ is $\Triple(f(\bar{z}), p, t)$ and the head of  $\sigma'$ is $\Triple(f(\bar{z}'), p, t')$, then $\bar{a} = h\circ\hpisigma(\bar{z}) = h\circ\hpisigma(\bar{z}')$, $b = h\circ\hpisigma(t)$ and $b' = h\circ\hpisigma(t')$.
  \end{enumerate}
\end{proposition}

We point out that Proposition~\ref{prop:existence-violation} identifies a necessary condition for the presence of some violation in the core pre-solution for a source instance $I$.
The condition is not sufficient for two reasons. First, $I$ is an instance of $\Rcal$ that does not necessarily satisfy the source functional dependencies.
Second, $b$ might be equal to $b'$.
Theorem~\ref{thm:pre-des-cone} adds sufficient conditions for handling these two missing cases.


\begin{theorem}
  \label{thm:pre-des-cone}
  These two statements are equivalent:
  \begin{itemize}
  \item For every instance $I$ of $\Rbold$, the core pre-solution for $I$ to $\E$ is value consistent.
  \item For every violation sort $(T,f,p)$, every $\pi$ s.t. $(T,f)$ is accessible in $\E$ with $\pi$, any two contentious st-tgds $\sigma,\sigma'$ of sort $(T,f,p)$, every $J$ solution for $\Instpisigma$ to $\Sigma_\fd$ it holds that $(\hpisigma \circ h)(t) = (\hpisigma \circ h)(t')$, where $t,t'$ are such that the head of $\sigma$ is $\Triple(f(\bar{z}), p, t)$ and the head of $\sigma'$ is $\Triple(f(\bar{z'}), p, t')$, and $h$ is the unique homomorphism from $\Instpisigma$ to $J$.
  \end{itemize}
\end{theorem}


\subsection{Node Kind Consistency}
\label{sec:ctwo}
Node kind inconsistency characterizes situations in which all chase sequences would fail due to the necessity of equating a literal and a non literal value by triggering a predicate functionality egd $T(x) \land \Triple(x, p, y) \land \Triple(x, p, y') \Rightarrow y = y'$ with homomorphism $h$ s.t. \emph{exactly one} among $h(y),h(y')$ is a literal.
In this case the corresponding chase sequence fails even if one of $h(y),h(y')$ is null.
This is a particularity of relational to RDF data exchange (in contrast to relational data exchange).

In the sequel we give a definition of node kind consistency and announce the propositions needed for proving the consistency theorem. 
The detailed definitions are rather technical and are presented in Appendix~\ref{app:proofs-consistency}.

For a typed graph $J$ we define the set $\Cotypes(J)$ of sets of \emph{types co-occurring} in all solutions $G$ of $J$ to $\E$ that include $J$.
That is, $X \in \Cotypes(J)$ if for any $G$ s.t. $J \subseteq G$ and $G \in \sol_\E(J)$, there exists a node $n$ in $G$ s.t. $X = \{T \in \Tcal \cup \{\Literal\} \mid T(n) \in G\}$.

\begin{definition}[Node kind consistent]
  \label{def:ctwo-instance}
  Let $I$ be a source instance and $J$ its core pre-solution to $\E$.
  $J$ is \emph{node kind consistent} if $\Cotypes(J)$ does not contain a set $X$ s.t. $\{\Literal, T\} \subseteq X$ for some type $T$ in $\Tcal$.
  The data exchange setting $\E$ is node kind consistent if for every $I$ instance of $\Rbold$, the core pre-solution for $I$ to $\E$ is node kind consistent.
\end{definition}

Node kind inconsistency is a sufficient condition for inconsistency.
\begin{lemma}
  \label{lem:not-ctwo-implies-inconsistent}
  For any $I$ instance of $\Rbold$, if the core pre-solution for $I$ to $\E$ is value inconsistent, then $I$ does not admit a solution to $\E$.
\end{lemma}
In Theorem~\ref{thm:pre-des-cone} we have shown that value inconsistency is another such sufficient condition.
The next lemma establishes that being value consistent and node kind consistent is a sufficient condition for $\E$ to be consistent.
\begin{lemma}
  \label{lem:cone-and-ctwo-implies-solution-exists}
  For any $I$ instance of $\Rbold$, if the core pre-solution for $I$ to $\E$ is value consistent and node kind consistent, then $I$ admits a solution to $\E$.
\end{lemma}

We are now ready to establish our main results regarding consistency of constructive data exchange settings.
The next theorem follows from Theorem~\ref{thm:pre-des-cone}, Lemma~\ref{lem:not-ctwo-implies-inconsistent}, Lemma~\ref{lem:cone-and-ctwo-implies-solution-exists}, and the fact that value consistency and node kind consistency are decidable.
\begin{theorem}[Consistency]
  \label{thm:consistency}
  $\E$ is consistent iff $\E$ is value consistent and node kind consistent.
\end{theorem}

Finally, we show that checking consistency of a constructive data exchange setting is co-NP complete.
The lower bound is shown using a reduction to the complement of SAT.
\begin{theorem}[Complexity of consistency]
  \label{thm:consistency-intractable}
  Checking consistency of a constructive relational to RDF data exchange setting
  is coNP-complete.
\end{theorem}


} 

\subsection{Non-Constructive st-tgds, Non-Deterministic Shape Schemas}
\label{sec:consistency-non-constructive}
The consistency checking algorithm can be extended to non-constructive data exchange settings but the lower co-NP complexity bound is not preserved by the extension.
Consider a data exchange setting $\E = (\Rbold, \Sbold, \Sigma_\st,\Fbold)$ with $\Sbold = (\Tcal, \delta)$ in which
the st-tgds in $\Sigma_\st$ can contain existential rules of the form $\varphi \Rightarrow \exists \bar{y}. \psi$ where function terms use only universally quantified variables.
We illustrate consistency checking on an example.

\begin{example}
  \label{ex:illustartion-acc-eq-rev}
  Consider shapes schema with types $T,U$ and rules $\delta(T,\p) = U^\ONE$ and $\delta(U,\q) = \Literal^\MAYBE$, and the st-tgds
  \begin{align*}
    &R(x,y,w) \Rightarrow \Triple(f(x),\p,g(y))\\
    &S(x',y') \Rightarrow \exists z'.T(f(x')) \land \Triple(f(x'),\p,z') \land \Triple(z',\q,y')\\
    &R(x'',y'',w'') \Rightarrow \Triple(g(y''),\q,w'')
  \end{align*}
  Even in presence of existential variables, we can statically infer that the st-tgd head atoms $\Triple(z',\q,y')$ and $\Triple(g(y''),\q,w'')$ are \emph{contentious}, then construct the source instance $I = \{R(x,y,z), S(x, y'), R(x'', y, z'')\}$ witness of value inconsistency of the data exchange setting at hand.
  Indeed, the core pre-solution to $I$ contains the facts $\{\Triple(f(x),\p,g(y)),$ $T(f(x)),$ $\Triple(f(x), \p, \bot_1), \Triple(\bot_1,\q,y'), U(\bot_1), \Triple(g(y),\q,w'')\}$.
  Triggering the predicate functionality rule for $\delta(T,p)$ we equate $\bot_1$ with $g(y)$.
  Then the last three atoms in $I$ constitute a violation of the predicate functionality rule for $\delta(U,q)$.

  The instance $I$ is discovered by exploring the possible interactions between the rules coming from the shape schema and the st-tgds' heads.
  We first remark that the terms $f(x)$ and $f(x')$ are \emph{equatable} (i.e. the target values produced by them might be equal as they are produced by the same IRI constructor), then type $T$ is \emph{accessible} for $f(x)$ due to the second st-tgd.
  The terms $g(y)$ and $z'$ are also equatable due to predicate functionality of $p$ for type $T$, and so are $g(y)$ and $g(y'')$ (same IRI constructor).
  Also, type $U$ is accessible for $g(y)$, so also for $z'$ and $g(y'')$ (type propagation of $\delta(T,p) = U^\ONE$).
  Thus the target atoms (generated during chase from) $\Triple(z',q,y')$ and $\Triple(g(y''),q,w'')$ can both have as subject the same value $g(y)$, and trigger a violation due to the predicate functionality $\delta(U,q)$.
  \qed
\end{example}

Similarly to the case of constructive data exchange settings, the consistency checking algorithm is based on the fact that $\E$ is value inconsistent iff there exists a value inconsistent instance $I$ among a finite set $V$ of source instances.
The latter is characterized by the presence of contentious atoms in st-tgds' heads, which in turn are discovered by a Datalog program.
Formal definitions and description of the algorithm are given in Appendix~\ref{app:consistency-non-constructive}, and allow us to establish 


\begin{theorem}
  \label{thm:consistency-non-constructive}
  Consistency is decidable for data exchange settings with existential st-tgds.
\end{theorem}
The exact complexity of the decision procedure is left for future work.

Finally we point out that if we consider non-deterministic shape schemas, then the complexity of checknig consistency increases.
The proof is given in Appendix~\ref{app:complexity-non-det}.

\begin{theorem}
  \label{thm:consistency-nondeterministic-schema}
  Checking consistency of a constructive relational to RDF data exchange setting with nondeterministic shape schema is $\Pi_2^p$-hard.
\end{theorem}


\section{Certain Query Answering}
\label{sec:query-answering}
In this section we investigate computing certain query answers. We focus mainly
on Boolean queries as it allows us to more easily present our constructions and
compare various classes of queries; later on we extend our results to
non-Boolean queries. Throughout this section we fix a constructive data exchange
setting $\E=(\Rbold,\Sbold,\Sigma_\st,\Fbold)$ and assume $\E$ is consistent. We
recall that for a Boolean graph query $Q$, $\mathit{true}$ is the \emph{certain
  answer} to a query $Q$ in $I$ w.r.t.\ $\E$ iff $\emph{true}$ is the answer to
$Q$ in every solution to $\E$ for $I$.

The standard approach to computing certain answers is to construct a universal
solution with the chase and evaluate the query against it (and if the query is
non-Boolean, we drop any answers that use null values). However, in the case of
consistent constructive relational to RDF data exchange, a finite universal
solution may not exists as it is the case in the example in Section
\ref{sec:example}. Indeed, the mutually recursive types \Bug and \User cause the
chase to loop ad infinitum: the user \textit{Edith} results in the node
\texttt{usr:2} of type $\User$ which required to track at least one
problem. Since in the relational database instance \textit{Edith} does not track
any bug, the chase needs to ``invent'' a fresh null IRI of type $\Bug$. This
node is required to have a user that has reported it and again the chase
``invents'' another fresh null IRI of type $\User$, and so on.

Instead, we construct a solution, where we avoid inventing nodes with the same
set of types, thus creating loops as illustrated in Figure~\ref{fig:rdf}. While
this solution is not universal, it seems quite natural, and interestingly, we
show that it has a different flavor of universality, one that can be captured
with the standard notion of graph simulation: any solution can be simulated in
it. We also show that this notion of universality is good enough for classes of
queries that are robust under simulation, and we identify a practical class of
forward nested regular expressions with this property. This yields a practical
class of queries with tractable consistent answers under data complexity. We
show that extending this fragment to full nested regular expression leads to
significant complexity increase. Finally, we also show that existing result on
chase with guarded tdgs and egds can be used to compute certain answers to
conjunctive queries.

\paragraph{Nested regular expressions} 
In this paper we focus mainly on the class of \emph{nested regular expressions}
(NREs) that have been proposed as the navigational core of
SPARQL~\cite{PEREZ2010255}. In essence, NREs are regular expressions that use
concatenation $\cdot$, union $+$, Kleene's closure $*$, inverse $-$, and permit
nesting and testing node and edge labels. Formally, NREs are defined with the
following grammar:
\[
E \coloncolonequals 
\epsilon \mid 
p \mid 
\Box \mid 
\langle \ell \rangle \mid 
[E] \mid 
E^* \mid 
E^- \mid 
E\cdot E \mid 
E+E 
\]
where $p\in\Pred$, $\ell\in\Iri\cup\Lit$, and $\Box$ is a distinguished wildcard
predicate symbol. An NRE is \emph{forward} (NRE$^\rightarrow$) if it does not
use the inverse operator. An NRE $E$ defines a binary relation $\sem{E}_G$ on
nodes of a graph $G$ as follows.
\begin{align*}
  &\sem{\epsilon}_G =\{(n,n)\mid n \in \nodes(G) \},
  &
  &\sem{[E]}_G=\{(n,n)\mid \exists m.\, (n,m)\in \sem{E}_G \},
  \\
  &\sem{p}_G=\{(n,m)\mid (n,p,m) \in G \},
  &
  &\sem{E_1+E_2}_G =\sem{E_1}_G \cup \sem{E_2}_G,
  \\
  &\sem{\Box}_G=\{(n,m)\mid \exists p\in \Iri. \, (n,p,m) \in G \},
  &
  &\sem{E_1\cdot E_2}_G=\sem{E_1}_G \circ \sem{E_2}_G,
  \\
  &\sem{\langle \ell \rangle}_G=\{(n,n)\mid n \in \nodes(G) \land n=\ell \},
  &
  &\sem{E^*}_G=\sem{E}_G^*,\qquad\sem{E^-}_G=\sem{E}_G^{-1}.
\end{align*}
An NRE $E$ is \emph{satisfied} in a graph $G$ iff $\sem{E}_G\neq\emptyset$.  We
point out that NREs are incompatible with conjunctive queries but even forward
NREs capture the subclass of acyclic conjunctive queries. Also, NREs (forward
NREs) properly captures 2-way regular path queries (regular path queries, resp.)

\subsection{Universal simulation solution}

\paragraph{Graph simulation and robust query classes}
We adapt the classic notion of graph simulation to account for null
values. Formally, a \emph{simulation} of a graph $G$ by a graph $H$ is a
relation $R\subseteq\nodes(G)\times \nodes(H)$ such that for any $(n,m) \in R$,
we have 1) $n$ is a literal node if and only if $m$ is a literal node, 2) if $n$
is not null, then $m$ is not null and $n=m$; and 3) for any outgoing edge from
$n$ with label $p$ that leads to $n'$ there is a corresponding outgoing edge
from $m$ with label $p$ that leads to $m'$ such that $(n',m')\in R$. The set of
simulations is closed under union, and consequently, there is always one maximal
simulation, and if $(n,m)$ is contained in it, we say that $n$ is simulated by
$m$.
Also, we say that $G$ is \emph{simulated} by $H$ if every node of $G$ is
simulated by a node of $H$. We are interested in simulation because it captures
the essence of exploring a graph by means of following outgoing edges only.
\begin{definition}
  A class $\mathcal{Q}$ of Boolean queries on graphs is \emph{robust under
    simulation} iff for any query $Q\in\mathcal{Q}$ and any two graph $G$ and
  $H$ such that $G$ is simulated by $H$, if $Q$ is true in $G$, then $Q$ is true
  in $H$. \qed
\end{definition}
The class of patterns presented above has this very property, which is shown
with an induction on the structure of the query. We point out, however, that our
approach is not restricted to forward NREs only.
\begin{lemma}
	\label{lem:regular-is-rub}
	The class of forward nested regular expressions is robust under simulation.
\end{lemma}
The related notion of bisimulation has found application in normalizing blank
nodes and essentially minimizing RDF graphs without altering its informational
contents~\cite{TzLaZe12}. Formally, a \emph{bisimulation} of a graph $G$ is a
simulation $R$ of $G$ by $G$ that is symmetric and reflexive. Again, there
exists a maximal bisimulation of any graph $G$, which we denote by
$\bisimulates$. We use the maximal bisimulation of a graph $G$ to construct it
\emph{reduct} $G/_\bisimulates$, which is the standard the quotient of the graph
$G$ and the equivalence relation $\bisimulates$ and replaces nodes of every
equivalence class by a single representative (details in
appendix). 
The main property that we employ in our proofs is that of the reduct of a typed
graph satisfies precisely the same shape schemas and the same queries
from any class robust under simulation, and furthermore it is the smallest typed
graph to have this property.

\paragraph{Universal simulation solution}
When dealing with classes of queries that are robust under simulation we employ
simulation instead of homomorphism to define a solution that allowing to find
all certain answers.
\begin{definition}
  A typed graph $\U$ is a \emph{universal simulation solution} to $\E$ for $I$
  iff $\U$ is simulated by every solution $J$ to $\E$ for $I$.  \qed
\end{definition}
And indeed, a universal simulation solution does allow us to capture certain
answers for queries from classes robust under simulation.
\begin{theorem}
	\label{thm:universal-satisfies-tree}
   Let $\mathcal{Q}$ be a class of Boolean graph queries robust under
   simulation. For any query $Q\in\mathcal{Q}$ and any consistent instance $I$ of
   $\Rbold$, true is the certain answer to $Q$ in $I$ w.r.t.\ $\E$ if and only
   if true is the answer to $Q$ in any universal simulation solution to $\E$
   for $I$.
\end{theorem}

The main challenge remains in constructing a universal simulation solution. We
begin with the \emph{core pre-solution} $J_0$ for $I$, which is the unique
minimal typed graph $J_0$ that satisfies the st-tgds $\Sigma_\st$ and the \TP
rules for $\Sbold$ (cf. Section~\ref{sec:preliminaries}).
The core pre-solution $J_0$ does not necessarily satisfy $\Sbold$ as it may
have \emph{frontier} nodes whose type requires outgoing edges that are
missing. To identify such nodes and add the necessary outgoing edges we first
identify the types associated to a node in a typed RDF graph
$\types_G(n) = \{T \mid T(n)\in G\}$. Also, we say that a type $T$
\emph{requires an outgoing $p$-edge} if $p\dbl S^\mu\in \delta(T)$ for some
$\mu\in\{\ONE,\PLUS\}$ and some type $S$, and by $\Req(X)$ we denote the set of
all IRIs that is required by any $T$ in $X$. Now, the frontier of $J_0$ is the
following set
\[
\Fbb = 
\{
(n,p) \mid n\in\nodes(J_0),\ p\in\Req(\types_{J_0}(n)),\ \nexists m.\,\Triple(n,p,m)\in J_0
\}.
\]
We also define a function that for a set of types $X$ satisfied at a node
indicates the set of types $\Delta(X,p)$ that must hold at any node reachable by
$p$-labeled edge
\[
\Delta(X,p) = \{S \mid \text{$p\dbl{}S^\mu\in\delta(T)$ for some $T\in X$ and $\mu\in\{\MAYBE,\ONE,\PLUS,\MANY\}$}\}
\]
Now, the set of additional null nodes that we add to $J_0$ is constructed in
an iterative process (where we identify each node with the set of types it is to
satisfy): $N=\bigcup_{i=0}^\infty N_i$, where
$N_0=\{\Delta(types(n),p) \mid (n,p)\in\Fbb\}$ and
$N_i=\{\Delta(X,p) \mid X \in N_{i-1},\ p\in\Req(X)\}$ for $i \geq 1$. Note that
we construct only subsets of the finite set of types $\Tcal$, and therefore,
this process eventually reaches a fix point. It may however be of size
exponential in the size of the schema, and in fact, with an elaborate example
using Chinese reminder theorem we can show that it is in fact a tight bound of
our construction. Now, the additional component of a universal simulation
solution is the following graph
\begin{multline*}
	G_\Sbold=\{(n,p,\Delta(\types_{J_0}(n),p))\mid (n,p) \in \Fbb \} \cup {} \\\{(X,p,\Delta(X,p)) \mid X\in N \land p\in \Req(X) \}\cup \{T(X) \mid X\in N \land T \in X \}.
\end{multline*}
We point out that $J_0\cup G_\Sbold$ does in fact satisfy $\Sigma_\st$ and
$\Sbold$ but it may not be the minimal universal simulation solution. However,
it suffices to take the bisimulation quotient of $G_\Sbold$ to ensure the
minimality: the constructed universal simulation solution is
$\U_0=J_0\cup G_\Sbold/_\bisimulates$. We point out that because $J_0$ does not
have any null  nodes, $\U_0=(J_0\cup G_\Sbold)/\bisimulates$.
\begin{theorem}
	\label{thm:quotient-solution-exp}
	For an instance $I$ of $\Rbold$, we can construct a size-minimal universal
   simulation solution $\U_0$ in time polynomial in the size of $I$ and
   exponential in the size of $\Sbold$. The size of $\U$ is bounded by a
   polynomial in the size of $I$ and an exponential function in the size of
   $\Sbold$.
\end{theorem}

\subsection{Complexity}
We can now characterize the data complexity of certain query answering. Recall
that data complexity assumes the query and the data exchange setting to be
fixed, and thus of fixed size, and only the source instance is given on the
input. Consequently, the size of universal simulation solution $\U_0$ is
polynomially-bounded by the size of $I$. Since the data complexity of evaluating
NREs is know to be PTIME~\cite{PEREZ2010255}, we get the following result.
\begin{theorem}
	\label{thm:complexity}
   The data complexity of computing certain answers to forward nested regular
   expressions w.r.t.\ constructive relational to RDF data exchange setting is
   in PTIME.
\end{theorem}

\paragraph{Full nested regular expressions}
Computing certain answers to the full class of NRE remains an open question. One
could explore using 2-way alternating automata (2ATAs) for infinite trees
corresponding to unraveling the universal simulation solution $\U_0$, a method
that has been successfully applied to the closely related problem of computing
certain answers to variants of regular path queries in the presence of
ontologies~\cite{CALVANESE201412,jung_et_al:LIPIcs:2018:8597}. However, using 2ATAs comes with significant computational
cost, and indeed, we show an increase in the complexity of computing certain
answers to NREs as compared to forward NREs. This increase is detected when we
fix the data exchange setting but consider both the query and the source
instance to be part of the input, a complexity measure that is between data and
combined complexity measures. Formally, for a class of Boolean graph queries
$\Q$ and a data exchange setting $\E$ we define the decision problem 
$D_\E^\Q =
\{
(I,Q) \mid 
Q\in\Q,\ \text{\textit{true} is the certain answer to $Q$ in $I$ w.r.t. $\E$}
\}$.
\begin{proposition}
  \label{prop:full-nre-intractable}
  For any constructive relational to RDF data exchange setting $\E$, 
  $D_\E^{\text{NRE}^\rightarrow}$ is in PTIME and $D_\E^{\text{NRE}}$ is
  PSPACE-hard.
\end{proposition}

\paragraph{Conjunctive queries}
The set of tgds in $\Sigma_\st\cup\Sigma_\Sbold$ is guarded and as such enables
using existing results by Cal\`{i} et al.~\cite{CALI201257} on tractability of
certain answering for conjunctive queries. We recall that the classes of
conjunctive queries and NREs are incomparable.
\begin{proposition}
	\label{prop:conjunctive-query-certain-answer}
   The data complexity of computing certain answers to conjunctive queries
   w.r.t.\ constructive relational to RDF data exchange setting is in PTIME.
\end{proposition}


\paragraph{Non-Boolean queries}
So far we have considered only Boolean queries and now we illustrate, on the
example of binary forward NREs, that the universal simulation solution can be
used to compute certain answers using the well-known method of evaluating the
query over $\U_0$ and dropping any answers using null values. Formally, a pair
of nodes $(n,m)$ is an \emph{answer} to an NRE $E$ in a graph $G$ iff
$(n,m)\in\sem{E}_G$. A pair $(n,m)$ is a \emph{certain answer} to an NRE $E$ in
$I$ w.r.t.\ $\E$ iff $(n,m)$ is an answer in every solution for $I$ to $\E$.
\begin{proposition}
  Given a constructive data exchange setting $\E$, a source instance $I$, and a
  forward NRE $E$, a pair $(n,m)$ is a certain answer to $E$ in $I$ w.r.t.\ $\E$
  if and only if $(n,m)$ is an answer to $E$ in a universal simulation
  solution $\U$ for $I$ w.r.t.\ $\E$ and neither $n$ nor $m$ are null. 
\end{proposition}
The above result can be generalized to any class of non-Boolean queries that is
robust under simulation. However, attempting to present a precise definition of
non-Boolean queries robust under simulation would exceed the space limits and we
leave it for the full version of the paper.


\section{Related Work} 
\label{sec:related}


R2RML is a W3C standard language for defining custom relational to RDF mappings
\cite{das:2011}, other languages such as YARRRML \cite{yarrrml} are
compiled to R2RML. These languages do not impose constraints on the target, and
consequently, the solution is always defined, unique, and trivially consistent
which makes the problems of consistency and certain query answering irrelevant.
In \cite{boneva:2015b}, Boneva et al. have studied relational to graph data
exchange with st-tgds and the target constraints based on conjunctions of nested
regular expressions. The framework is incomparable to the framework presented in
this paper.

Viewing RDF as a ternary relation and expressing shape constraints with a set of
target dependencies, thus reducing our framework to the standard relational data
exchange~\cite{fagin:2005a}, allows us to
translate back existing results but only to a certain degree. Most notably, the
work on chase with guarded tgds~\cite{CALI201257} allows us to show that
computing certain answers to conjunctive queries in our framework is tractable
(Proposition~\ref{prop:conjunctive-query-certain-answer}). In general, other
works consider dependencies that are unsuitable to capture our mappings and
shape schemas, focus on query classes that are not as well suited to query RDF
as are NREs, or being very generic incur a much higher computational cost. For
instance, data exchange with weakly acyclic tgds and edgs~\cite{fagin:2005a}
is suitable for capturing only a restricted \emph{weakly-recursive} shape
schemas \cite{boneva:18} that do not result in an infinite chase. While there exist works on data
exchange that consider queries that go beyond conjunctive queries and add
elements of transitive closure, such as XML tree
patterns~\cite{murlak-xml-absolute-consistency,amano2014} or Datalog
fragments~\cite{arenas2011}, they come at a price of high complexity. Also, while
shape schemas are reminiscent of DTDs (or more closely of XML Schemas), XML is
an ordered model and the source to target mappings in XML data exchange need to
specify the relative order of elements or a universal solution may fail to
exists, and even if unordered XML is employed computing certain answers easily
becomes intractable~\cite{arenas2008}. Finally, there is work on answering classes of
regular path queries in description
logics~\cite{CALVANESE201412,bienvenu2015} allows to easily capture our constructive data
exchange settings and the considered classes of queries seem suitable for
querying RDF but again they come with significant computational cost. However,
we believe that the underlying use of 2-way alternating tree automata
(2ATA)~\cite{CGKV88,THB95} can be employed to computing certain answers to NREs
in our framework, which we intend to pursue in our future work.

\section{Conclusion and Future Work}
\label{sec:future}
We have presented a data exchange framework for exporting in a R2RML-like
fashion a relational database to RDF with (non-overlapping) IRI constructors and
target shape schema. We have studied the problems of consistency and have shown
it to be coNP-complete using an intricate characterization. We have also studied
computing certain answers to forward nested regular expressions and shown it be
tractable using a novel construction of universal simulation solution. We have
also shown that extending the framework in a number of natural directions
generally leads to an increase of complexity.


Future research directions include a complete complexity analysis of relational
to RDF data exchange with non-constructive st-tgds and nondeterministic and
disjunctive shape schemas, and exploring using 2ATA for computing certain query
answers to the full fragment of nested relational expressions.


\bibliography{exchange}

\newpage
\appendix
 \includecomment{fullproof}
 \excludecomment{sketch}
\section{Omitted Formalisms and Proofs}
\newcommand{\subsubsubsection}[1]{\paragraph{#1}\mbox{}\\}
\setcounter{secnumdepth}{4}
\setcounter{tocdepth}{4}
\newcommand{\Dom}{\mathsf{Dom}}
\newcommand{\dom}{\mathit{dom}}
\newcommand{\eqClass}{\mathit{eq\text{\sf-}class}}
\newcommand{\NullSet}{\mathsf{Null}}
\newcommand{\ConstSet}{\mathsf{Const}}
\newcommand{\head}{\mathit{head}}
\newcommand{\vars}{\operatorname{\mathit{vars}}}
\newcommand{\body}{\mathit{body}}
\newcommand{\Vars}{\mathsf{Vars}}
\newcommand{\Model}{M}
\newcommand{\ConstLit}{\mathsf{ConstLit}}
\newcommand{\pe}{\PE}
\paragraph{Constants, nulls, and variables} 
We assume a fixed enumerable domain $\Dom$. For the purposes of this paper, we
assume the domain to be partitioned into three infinite subsets
$\Dom = \Iri\cup \Lit \cup \NullIri$ of \emph{IRIs}, \emph{literals}, and
\emph{blank} node identifiers respectively. We assume an infinite set of
\emph{null literals} $\NullLit \subseteq \Lit$ and identify the set of
\emph{null values} $\NullSet=\NullLit\cup\NullIri$. We refer to the
remaining elements as \emph{constants} $\ConstSet = \Dom \setminus \NullSet$,
and additionally, identify non-null literals
$\ConstLit=\Lit\setminus\NullLit=\ConstSet\cap\Lit$. Also, we fix an
infinite set of (first-order) variables $\Vars$. In the sequel, we use $a$, $b$,
and $c$ to range over elements of $\Dom$ and $\bar{a}$, $\bar{b}$, $\bar{c}$
to range over sequences of elements of $\Dom$. Similarly, we use $x$, $y$, and
$z$ to range over variables and $\bar{x}$, $\bar{y}$, and $\bar{z}$ to range
over sequences of variables.

\paragraph{First-order logic} We recall basic notions of first-order logic. We
assume a finite \emph{signature}, which consists of a set of relation symbols
$\Rcal$ and a set of function symbols $\Fcal$, each symbol having a fixed
arity. We use first-order formulas with relation and function names from
$\Rcal\cup\Fcal\cup\{\mathord=\}$, variables from $\Vars$, and constants from
$\Dom$, however, we only employ \emph{flat terms} that do not use nested 
applications of function symbols. An \emph{atom} has the form $R(\bar{t})$,
where $R\in\Rcal$ and $\bar{t}$ is a sequence of terms. A \emph{relational
  atom} does not use any function symbols. A \emph{clause} is a conjunction of
atoms and we often view it as a set of atoms.  A formula is \emph{closed} if it
has no free variables. A formula is \emph{ground} if it uses no variables
whatsoever. A \emph{fact} is a ground relational atom. A clause is
\emph{relational} if it employs only relational atoms.

A \emph{structure} (or a \emph{model}) $\Model$ (over the signature
$\Rcal\cup\Fcal$) is a mapping that associates to every relation and function
symbol $\xi\in\Rcal\cup\Fcal$ a corresponding relation or function $\xi^\Model$
on elements of $\Dom$ of appropriate arity. The semantics of a first-order logic
formula $\varphi$ over a model $\Model$ is captured with the \emph{entailment
  relation} $\Model\models\varphi$ defined in the standard fashion.  The
entailment relation is extended to a set of formulas: $\Model\models\Phi$ iff
$\Model\models\varphi$ for every $\varphi\in\Phi$. Also, we often view a model
over a signature consisting of relation symbols only as the set of all facts
satisfied by the model.

\paragraph{Dependencies} A \emph{dependency} $\sigma$ is a closed first-order
formula of the form $\forall\bar{x}.\varphi\Rightarrow\exists\bar{y}.\psi$,
and we define $\body(\sigma)=\varphi$, $\head(\sigma)=\psi$, and
$\vars(\sigma)=\bar{x}\cup\bar{y}$. $\sigma$ is an \emph{equality-generating
  dependency} (egd) if $\varphi$ is a clause and $\psi$ is a conjunction of
equality conditions $x=y$ on pairs of variables. $\sigma$ is a
\emph{tuple-generating dependency} (tgd) if both $\varphi$ and $\psi$ are
clauses. A tgd is \emph{full} if it uses no existentially quantified variables
$\bar{y}$. The use of the equality relation $\mathord=$ in an egd $\sigma$
implies a binary relation on the variables of $\sigma$, its (reflexive,
symmetric, and transitive) closure gives an equivalence relation that identifies
variables that need to have the same value, and by $\eqClass(\sigma)$ we denote
the set of equivalence classes to which this relation partitions the variables
$\bar{x}$ (the variables $\bar{y}$ can be ignored).

\paragraph{Relational databases} A \emph{relational schema} is a pair
$\Rbold = (\Rcal, \Sigma_\fd)$ where $\Rcal$ is a set of relation names, each with a
fixed arity, and $\Sigma_\fd$ is a set of \emph{functional dependencies} (fds) of the
form $R:X\rightarrow Y$, where $R\in\Rcal$ is a relation name of arity $n$, and
$X,Y\subseteq\{1,\ldots,n\}$. An fd $R:X\rightarrow Y$ is a short for the egd
$\forall\bar{x},\bar{y}.\ R(\bar{x})\land R(\bar{y})\land
\textstyle\bigwedge_{i\in X} (x_i=y_i) \Rightarrow \textstyle\bigwedge_{j\in Y}
(x_j=y_j)$. An \emph{instance} of $\Rbold$ is a model $I$ over $\Rcal$, and
unless we state otherwise, in the sequel we work only with instances that use
literal constants from $\ConstLit$. The \emph{active
  domain} $\dom(I)$ of the instance $I$ is the set of elements of $\Dom$
used in $I$.

\paragraph{Graphs} An \emph{RDF graph} (or simply a \emph{graph}) is a finite
set $G$ of triples in
$(\Iri \cup \NullIri) \times \Pred \times (\Iri \cup \NullIri \cup
\Lit)$. We view $G$ as an edge labeled graph by interpreting a triple
$(s,p,o)$ as a $p$-labeled edge from the node $s$ to the node $o$. The set of
\emph{nodes} of $G$, denoted $\nodes(G)$, is the set of elements that appear on
first or third position of a triple in $G$.

\paragraph{Shape constraints as dependencies}
A deterministic shapes schema $\Sbold$, as defined in Section~\ref{sec:preliminaries}, can be captured with a set
$\Sigma_\Sbold$ of egds and tgds. $\Sigma_\Sbold$ contains for all type
$T\in\Tcal$ and all shape constraint $\delta(T,p)=S^\mu$
\begin{itemize}
\item $\tp(T,p,S)=\forall x,y.\ T(x)\land\Triple(x,p,y)\Rightarrow S(y)$,
\item $\pf(T,p)=\forall x.\ T(x)\Rightarrow \exists y.\ \Triple(x,p,y)$,
  if $\mu\in\{\ONE,\PLUS\}$,
\item
  $\pe(T,p)=\forall x,y,z.\
  T(x)\land\Triple(x,p,y)\land\Triple(x,p,z)\Rightarrow y=z$, if
  $\mu\in\{\ONE,\MAYBE\}$.
\end{itemize}
It is easy to see that a typed graph $(G,\typing)$ satisfies
$\Sbold$ if and only if $(G,\typing)\models\Sigma_\Sbold$.

\paragraph{Homomorphisms and universal solutions} A substitution is a function
$h:\Dom\cup\Vars\rightarrow\Dom\cup\Vars$ that is different from identity on a
finite set $\dom(h)$ of null values and variables, and furthermore, $h$ assigns
a value in $\Dom$ to every element in $\dom(h)$. We assume that the library of
IRI constructors is known from the context, and extend substitutions to flat
terms while applying interpretations of the IRI constructors:
$h(f(t_1,\ldots,t_k))=\smash{f^F}(h(t_1),\ldots,h(t_k))$. We further extend
homomorphisms, in a standard fashion, to atoms
$h(R(t_1,\ldots,t_k))=R(h(t_1),\ldots,h(t_k))$, and to sets of atoms
$h(A)=\{h(\alpha) \mid \alpha\in A\}$. Recall that both instances and clauses
can be viewed as set of atoms. Now, a \emph{homomorphism} of $I_1$ in $I_2$ is a
substitution $h$ such that $h(I_1)\subseteq I_2$. A homomorphism $h'$
\emph{extends} a homomorphism $h$, written $h\subseteq h'$, if
$\dom(h)\subseteq\dom(h')$ and $h'(x)=h(x)$ for all $x\in\dom(h)$. A
\emph{universal} solution $U\in\sol_\E(I)$ is a solution that subsumes all
other solutions i.e., for any $J\in\sol_\E(I)$ there is a homomorphism of $U$
in $J$.

\paragraph{Chase} We recall the \emph{chase} procedure for tgds and egds. Let
$\sigma=\forall\bar{x}.\ \varphi\Rightarrow\exists \bar{y}.\ \psi$. If
$\sigma$ is a tgd, we say that it is \emph{triggered} in $I$ by $h$ if
$\dom(h)=\bar{x}$, $h(\varphi)\subseteq I$, and there is no extension $h'$ of
$h$ such that $h'(\psi)\subseteq I$. It has a successful \emph{execution} $h'$
yielding $I'$, in symbols $\smash{I \xrightarrow{\sigma,h'} I'}$, if $h'$ is an
extension of $h$ such that $\dom(h')=\bar{x}\cup\bar{y}$ and
$I'=I\cup h'(\psi)$.

Next, suppose that $\sigma$ is an egd. It is \emph{triggered} in $I$ by $h$ if
$\dom(h)=\bar{x}$, $h(\varphi)\subseteq I$, and there is
$\bar{z}\in\eqClass(\sigma)$ and $z_1,z_2\in\bar{z}$ such that $h(z_1)\neq
h(z_2)$. It has a successful \emph{execution} $I'$ with $h'$, in symbols
$\smash{I \xrightarrow{\sigma,h'} I'}$, if $I'=h'(I)$ and $h'$ is a homomorphism
such that $\dom(h')=h(\dom(h))\cap\NullSet$ i.e., $h'$ assigns values to the
null values used by $h$, and for any $\bar{z}\in\eqClass(\sigma)$ and any
$z_1,z_2\in\bar{z}$ we have $h'(h(z_1))=h'(h(z_2))$. If $\sigma$ is triggered
in $I$ by $h$ but does not have a successful execution, we say that it
\emph{fails}, in symbols $\smash{I \xrightarrow{\sigma,h} \bot}$.

Now, a \emph{chase} sequence on $I_0$ with a set $\Sigma$ of tgds and egds is a
possibly infinite sequence
$\smash{I_0 \xrightarrow{\sigma_0,h_0} I_1 \xrightarrow{\sigma_1,h_1} I_2}
\ldots$, where $\sigma_i\in\Sigma$ for all $i$. A \emph{terminating} chase
sequence ends with a failure or an instance that triggers no dependency in
$\Sigma$. It is a classic result that a universal solution $U$ for
$I$ to $\Sigma$ exists if and only if there is a terminating chase sequence on $I$ with
$\Sigma$ that ends with $U$ \cite{fagin:2005a}. Naturally, this result extends
to constructive data exchange settings with fixed IRI constructors.

\subsection{Proofs for Section~\ref{sec:consistency-des} (Consistency)}
\label{app:proofs-consistency}
\newcommand{\Cotypes}{\operatorname{\mathit{CoTypes}}}
\newcommand{\sttc}{\Sigma_\st \cup \Sigma_\Sbold^\tp }
\newcommand{\shexmostonedep}{\Sigma_\Sbold^\pf}
\newcommand{\shexdep}{\Sigma_\Sbold}

\newcommand{\T}{\ensuremath{\mathbf{t}}}
\newcommand{\F}{\ensuremath{\mathbf{f}}}
\newcommand{\ran}{\mathit{ran}}
\newcommand{\chasestep}[4]{#1\; \smash{\xrightarrow{#2,\,#3}}\; #4}
\newcommand{\Instpisigma}{I_{\pi,\sigma,\sigma'}}
\newcommand{\Bpisigma}{B_{\pi,\sigma,\sigma'}}
\newcommand{\hpisigma}{h_{\pi,\sigma,\sigma'}}
We start by precising the definition of a chase step in order to take into account the possible conflict due to merging a literal and a non-literal node.
Suppose that $\sigma$ is an egd. It is \emph{triggered} in $I$ by $h$ if
$\dom(h)=\bar{x}$, $h(\varphi)\subseteq I$, and there is
$\bar{z}\in\eqClass(\sigma)$ and $z_1,z_2\in\bar{z}$ such that $h(z_1)\neq
h(z_2)$. It has a successful \emph{execution} $I'$ with $h'$, in symbols
$\smash{I \xrightarrow{\sigma,h'} I'}$, if 
\begin{itemize}
\item[(1)] any $\bar{z}\in\eqClass(\sigma)$ and $z_1,z_2\in\bar{z}$ such that $h(z_1)\neq
h(z_2)$ satisfy $h(z_1),h(z_2)$ are both literals, or $h(z_1),h(z_2)$ are both non literals, and
\item[(2)] $I'=h'(I)$ and $h'$ is a homomorphism
such that $\dom(h')=h(\dom(h))\cap\NullSet$ i.e., $h'$ assigns values to the
null values used by $h$, and for any $\bar{z}\in\eqClass(\sigma)$ and any
$z_1,z_2\in\bar{z}$ we have $h'(h(z_1))=h'(h(z_2))$.
\end{itemize}
If $\sigma$ is triggered
in $I$ by $h$ but does not have a successful execution, we say that it
\emph{fails}, in symbols $\smash{I \xrightarrow{\sigma,h} \bot}$.

\subsubsection{Value consistency}
We prove here the different propositions made in Section~\ref{sec:cone}.

\subsubsubsection{Proof of Lemma~\ref{lem:existence-type-fact}}
\begin{fullproof}
We show the left-to-right direction. 
We claim that (1) for any instance $I$ of $\Rbold$, a solution for $I$ to $\Sigma_\st\cup\Sigma_\Sbold^\tp$ is included in a solution for $I$ to $\E$.

We now show (1). Take a universal solution $J$ to $\Sigma_\st\cup\Sigma_\Sbold^\tp$ and a universal solution $J'$ to $\E$. We prove that $J\subseteq J'$. We fix two chase sequence $s$ and $s'$ such that the instance where there can not be triggered more rules are $J$ and $J'$ for $s$ and $s'$ respectively. Since $s'$ is finite, i.e. there is not failure, then the egds only are triggered for those triples that contain null as objects. This is not produced by $\Sigma_\Sbold^\tp$. Since $s'$ has applied the rules in $\Sigma_\st\cup \Sigma_\Sbold^\tp$, then $J$ is in $J'$.

Now, we take a pair $(T_n,f_n)\in \Tcal\times \Fcal$. Assume $(T_n,f_n)$ is accessible in $\E$. Then, we construct an instance $I$ of $\Rcal$ where $\Rcal=\{R_1,\ldots, R_n\}$ such that $\dom(I)=\{b\}$ and for each $R\in \Rcal$ is of arity $n\in\mathbb{N}$. Because $(T_n,f_n)$ is accessible, we know that there is a sequence $\sigma_0,\ldots, \sigma_i$ of st-tgds s.t.:
\begin{itemize}
	\item $\head(\sigma_0) = T_0(f_0(\bar{y_0}))$, and
	\item $\head(\sigma_i) = \Triple(f_{i-1}(\bar{x_i}), p_i, f_i(\bar{y_i}))$ for any $1 \le i \le n$, and
	\item $(p\dbl{} T_i)^\mu \in \delta(T_{i-1})$ for some multiplicity $\mu$, and
	\item $(T,f,p) = (T_n,f_n,p_n)$
\end{itemize}
for some type symbols $\{T_i \mid 0 \le i < n\} \subseteq \Tcal$, function symbols $\{f_i \mid 0 \le i < n\} \subseteq \Fcal$, and IRIs $\{p_i \mid 1 \le i < n\} \subseteq \Iri$.

Now, we take any solution $J$ for $I$ to $\E$. By chasing $I$ with the sequence of st-tgds  $\sigma_0,\ldots, \sigma_i$ and $\Sigma_\Sbold^\tp$ of triple constraints from accessibility of $(T_n,f_n)$, we have that
\begin{align*}
A=\{
T(f_0(\bar{b})),\Triple(f_0(\bar{b}),p_1,f_1(\bar{b})),T_1(f_1(\bar{b})),\Triple(f_1(\bar{b}),p_2,f_2(\bar{b})),\ldots,\\\Triple(f_{n-1}(\bar{b}),p_n,f_n(\bar{b})),T_n(f_n(\bar{b}))\} 
\end{align*}

Since chase sequence of $\Sigma_\st\cup \Sigma_\Sbold^\tp$ is finite, then there is an instance $A$ where there is no rule triggered. By chase sequence definition, $A$ is a universal solution to $\Sigma_\st \cup \Sigma_\Sbold^\tp$. By definition of universal solution, there is an homomorphism $h:A\to J'$ such that $h(c)=c$ for all $c \in \dom(A)$ and $J'$ a solution to $\Sigma_\st \cup \Sigma_\Sbold^\tp$. Since in $A$ there are no nulls, $A \subseteq J'$. By claim (1) and $A \subseteq J'$, $A \subseteq J$. Let $a=f_n^F(\bar{b})$, i.e. $a\in \ran(f_n^F).$ Then $T_n(a)\in J$. Thus, $T_n(a)$ is in all solutions for $I$ to $\E$.

Now, we show the right-to-left direction. We take a pair $(T_n,f_n)$ and assume there exist an instance $I$ of $\Rcal$ and a constant $a$ in $\ran(f_n^F)$ s.t. $T_n(a)$ is a fact in all solutions for $I$ to $\E$. 

We have to prove that $(T_n,f_n)$ is accessible in $\E$. We fix a chase sequence $s=I=J_0\xrightarrow{\sigma_1,h_1}J_1\xrightarrow{\sigma_2,h_2}\ldots J_{m-1}\xrightarrow{\sigma_m,h_m}J_m$ where $\sigma_i$ a dependency in $\Sigma_\st$ and $h_i:\sigma_i\to J_i$ an homomorphism for any $i \in \{ 1,\ldots,m\}$. 

We claim (1) that for any finite chase sequence $s$ for dependencies $\Sigma_\st\cup \Sigma_\Sbold^\tp$, any $a \in \ConstSet$, any $T\in \Tcal$, any $k\in \{1,\ldots,|s|\} $, if $J_k$ contains the fact $T(a)$ then there is $f\in \Fcal$ such that $(T,f)$ is accessible in $\E$.

We claim (2) that for any finite chase sequence $s$ for dependencies $\Sigma_\st\cup \Sigma_\Sbold^\tp$, any $a \in \ConstSet$, any $T\in \Tcal$, any $k\in \{2,\ldots,|s|\} $, if $J_k$ contains the fact $T(a)$ then there is $k'<k$ and $T'\in \Tcal$ and $p\in \Pred$ and $a'\in \ConstSet$ such that $J_{k'}$ contains the facts $T'(a'),\Triple(a',p,a)$ and $p\dbl{}T''^\mu\in\delta(T')$ for some $\mu$.

We claim (3) for any finite chase sequence $s$, any $k\in \{1,\ldots,|s|\}$, any $b,b'\in \ConstSet$ and any $q\in \Iri$ if $J_k$ contains the fact $\Triple(b,q,b')$ then there is a st-tgd $\sigma\in \Sigma_\st$ such that $\head(\sigma)=\Triple(f(\bar{y_1}), q, g(\bar{y_2}))$ for some $f,g \in \Fcal$ and $\bar{y_1}$ and $\bar{y_2}$ in $\vars(\sigma)$; and there is an homomorphism $h:\head(\sigma)\to J_{k}$ such that $h(\bar{y_1})=b$ and $h^F(\bar{y_2})=b'$.

We now prove claim (1). Let $\sigma$ be a rule in $\Sigma_\st$ such that $\head(\sigma)=T(f(\bar{x}))$ for some $\bar{x}\in \vars(\sigma)$. Let $h:\sigma\to J_k$ such that $h(f(\bar{x}))=a$. Since $J_k$ contains a fact produced in chase step of $s$, then this $h$ was triggered in some instance before $J_k$. By definition of accessibility, we conclude that $(T,f)$ is accessible in $\E$.

Next, we prove claim (2). Let $\sigma$ be a rule in $\Sigma_\Sbold^\tp$ such that $\head(\sigma)=T(x)$. Let $h:\sigma\to J_k$ such that $T(h(x))=a$. Since $J_k$ contains a fact produced in chase step of $s$, then $\sigma$ exists with $h$ because it was triggered in some instance before $J_k$. Let that instance be $J_{k'}$ such that $k'<k$. Thus, $T'(h(x)),\Triple(h(x),p,h'(y))$ in $J_{k'}$ for some $p\in \Iri$. Let $h(x) =a'$ for some $a'\in \ConstSet$. Since $\sigma\in \Sigma_\Sbold^\tp$ then $p\dbl{}T\in \delta(T')$.

Finally, we prove claim (3). Assume $J_k$ contains the fact $\Triple(b,q,b')$. Since $J_k$ is an instance of chase sequence at $k$ chase step of $\Sigma_\st\cup \shexdep$, then the only rule that can be triggered in some instance before $J_k$ is of the form $\varphi \Rightarrow \Triple(f(\bar{y_1}), q, g(\bar{y_2}))$ for some $\bar{y_1}$ and $\bar{y_2}$ in $\vars(\varphi)$. Thus, we conclude that there is an homomorphism $h$ such that $h(f(\bar{y_1}))=b$ and $h^F(q)=q$ and $h(f(\bar{y_2}))=b'$.

Since the fact $T_n(a)\in J_m$ is the result of either trigger a rule in $\Sigma_\st$ or $\Sigma_\Sbold^\tp$, we can apply claim (1) or (2) respectively. Considering the  rule be in $\Sigma_\st$ and by claim (1), we have that $(T_n,f_n)$ is accessible. Considering the rule be in $\Sigma_\Sbold^\tp$, we obtain that there is a $k'<m$ and $T'\in \Tcal$ and $a'\in \ConstSet$ and $p\in \Pred$ such that the facts $T'(a')$ and $\Triple(a',p,a)$ are in $J_{k'}$ and $p\dbl{}T_n^\mu\in\delta(T')$ for some $\mu$. Applying claim (3) to the fact $\Triple(a',p,a)\in J_{k'}$, we have that $\sigma=\Triple(f(\bar{y_1}), p, f_n(\bar{y_2}))$ for some $f\in \Fcal$ such that $a'\in \ran(f^F)$. Let this $f$ be $f_{n-1}$. We analyze the fact $T'(a')$ as in the beginning with claims (1) or (2). Assume that claim (1) is not applied until we are in step $2$. Then we have only applied claim (2) and (3) sequentially obtaining at this chase step that $T_0(a_0),\Triple(a_0,p_1,a_1)$ in $J_2$. By applying claim (1) to the fact $T_0(a_0)$ we obtain a rule $\sigma_0=T_0(f_0(\bar{y_0}))$. This rule together with the set of rules $\sigma_i,\ldots,\sigma_j$ where $i<j<m$ that are in $\Sigma_\st$, which were obtained by the application of claim (3) allows to conclude that $(T_n,f_n)$ is accessible in $\E$.
\end{fullproof}
\begin{sketch}
  \begin{proof}[Sketch of proof.]
    For the left-to-right direction, let $b \in \Dom$.
    Consider the instance $I$ that contains exactly one fact $R(b^{\mathit{arity(R)}})$ for any relational symbol $R$ in $\Rcal$, where $b^n$ is the $n$-tuple containing only $b$'s.
    Then we show that $J_0$, the core pre-solution for $I$ to $\E$, contains the fact $T(f^F(b^{\mathit{arity}(f)}))$ whenever $(T,f)$ is accessible in $\E$.

    For the right-to-left direction, we fix an arbitrary terminating chase sequence $s$ on $I$ with $\Sigma_\st \cup \Sigma_\Sbold^\tp$ producing the core pre-solution $J_0$.
    Using that $T(a)$ is a fact in $J_0$, we show that $s$ necessarily contains chase steps with dependencies as those Definition~\ref{def:accessible}, thus witnessing that $(T,f)$ is accessible in $\E$.
  \end{proof}
\end{sketch}

\subsubsubsection{Proof of Proposition~\ref{prop:existence-violation}}
\begin{fullproof}
Suppose first that $\pi, \sigma, \sigma', h$ are as in the premise of \ref{item:prop-existence-violation:1}. and let $\pi = \sigma_0, \ldots, \sigma_n$, $\head(\sigma) = \Triple(f(\bar{z}), p, t)$, $\head(\sigma') = \Triple(f(\bar{z'}), p, t')$ and the $\sigma_0, \ldots, \sigma_n$  as in Definition~\ref{def:accessible}, so $(T,f,p) = (T_n,f_n,p_n)$.
Because $(T_n,f_n)$ is accessible in $\E$ with $\pi$, we know that $\Sigma_\Sbold^\tp$ contains the rules $\tp(T_{i-1},p_i,T_i)$ for any $0 < i \le n$. 
  Let $\tp(T_{i-1},p_i,T_i) = T_{i-1}(u_i) \land \Triple(u_i, p_i, v_i) \Rightarrow T_i(v_i)$ for any $0 < i \le n$, where w.l.g. $u_i,v_i$ are fresh w.r.t. the variables used in $\sigma_0,\ldots,\sigma_n,\sigma,\sigma'$ and $\{u_i,v_i\}$ is disjoint from $\{u_j,v_j\}$ whenever $i\neq j$.
Let $h' = h\circ\hpisigma$, then by definition of $\Instpisigma$ and $h$ it follows that $I$ is the disjoint union of $h'(\Bpisigma)$ and $I'$, the latter containing the facts of $I'$ that are not images by $h$ of some fact in $\Instpisigma$.

Consider the following chase sequence $s$ starting at $I$; note that in the sequel we abuse the notation and use $h'$ as its restriction on any subset of variables in its domain.
The first chase step of $s$ is $\chasestep{I}{\sigma_0}{h'}{I_0}$.
The subsequent chase steps are defined inductively by adding the following two chase steps for all $0 < i \le n$:
$$
\chasestep{I_{i-1}}{\sigma_i}{h'}{I'_i} \chasestep{}{\tp(T_{i-1},p_i,T_i)}{h_i}{I_i},
$$
where $h_i$ is defined by $h_i(u_i) = h'(f_{i-1}(\bar{y}_{i-1}))$ and $h_i(v_i) = h'(f_i(\bar{y}_i))$.

Thus $s$ is of the form:
$$
\chasestep{I}{\sigma_0}{h'}{I_0} \chasestep{}{\sigma_1}{h'}{I'_1} \chasestep{}{\tp(T_0,p_1,T_1)}{h_1}{I_1} \rightarrow \cdots \rightarrow \chasestep{I_{n-1}}{\sigma_n}{h'}{I'_n} \chasestep{}{\tp(T_{n-1},p_n,T_n)}{h_n}{I_n}.
$$
We now show that $s$ is indeed a chase sequence. 
That is, we need to show that the homomorphism of each step above is indeed a homomorphism from the body of the dependency being applied to the instance to which the step is applied.
It immediately follows from the definitions and hypotheses that (1) $I_0 = I' \cup h'(\Bpisigma) \cup T_0(h(f_0(\bar{y}_0)))$ where $I'$ contains the facts of $I$ that are not images of some fact of $\Instpisigma$ by $h$.
For any $1 \le i \le n$ we show the following by induction on $i$:
\begin{itemize}
\item[(2)] $I'_i = I_{i-1} \cup h'(\head(\sigma_0) \cup \cdots \cup \head(\sigma_{i}))$;
\item[(3)] $I_i = I'_i \cup T_i(h'(f_i(\bar{y}_i)))$.
\end{itemize}
For the base case $i=1$. 
From (1) it follows that $h':\sigma_1 \to I_0$ is a homomorphism, and by definition of the chase, applying this homomorphism on $I_0$ yields $I'_1 = I_0 \cup h'(\head(\sigma_1))$, thus (2) holds.
Now from (1) and (2) we know that $I'_1$ contains the facts $T_0(h'(f_0(\bar{y}_0)))$ and $\Triple(h'(f_0(\bar{x}_1)), p_1, f_1(\bar{y}_1)) = h'(\head(\sigma_1))$.
Recall that by definition, $\hpisigma(\bar{x}_1) = \hpisigma(\bar{y}_0)$, so also $h'(\bar{x}_1) = h'(\bar{y}_0)$, thus $h_1$ is indeed an homomorphism from $T_{i-1}(u_i) \land \Triple(u_i, p_i, v_i)$ into $I'_i$ and the resulting instance is indeed $I'_i \cup T_i(h'(f_i(\bar{y}_i)))$.

The same arguments apply for the induction step for showing that $h':\sigma_i \to I_{i-1}$ and $h_i: \tp(T_{i-1},p_i,T_i)$ are homomorphisms, and their application yields the instances described in (2) and (3).

Consider now the chase sequence
$$
s' = \chasestep{I_n}{\sigma}{h}{I_\sigma} \chasestep{}{\sigma'}{h}{I_{\sigma'}}.
$$
It immediately follows from the definition of $h$, from (3) and from the definition of a chase step that $I_\sigma = I_n \cup \{\Triple(h(f_n(\bar{y}_n)), p_n, h(t))\}$ and $I_{\sigma'} = I_\sigma \cup \{Triple(h(f_n(\bar{y}_n)), p_n, h(t'))\}$.

Finally, consider any terminating chase sequence by $\Sigma_\st \cup \Sigma_\Sbold^\tp$ starting at $I_{\sigma'}$, and let $J'$ be its terminal instance; we know that such finite chase instance exists because the dependencies in $\Sigma_\st \cup \Sigma_\Sbold^\tp$ are full.
Then $J'$ is a universal solution for $I$ to $\Sigma_\st \cup \Sigma_\Sbold^\tp$ and moreover $I_{\sigma'} \subseteq J'$, so $\{T_n(h'(f_n(\bar{y}_n))), \Triple(h'(f_n(\bar{y}_n)), p_n, h'(t)), \Triple(h'(f_n(\bar{y}_n)), p, h'(t'))\} \subseteq J' \subseteq J$ for any $J$ solution for $I$ to $\Sigma_\st \cup \Sigma_\Sbold^\tp$.
We conclude the proof of the left-to-right direction of \ref{item:prop-existence-violation:1} remarking that by definition, $\hpisigma(\bar{y}_n) = \hpisigma(\bar{z}) = \hpisigma(\bar{z}')$, so also $h'(\bar{y}_n) = h'(\bar{z}) = h'(\bar{z}')$.
This also shows \ref{item:prop-existence-violation:2}. in the case where the left-to-right direction of \ref{item:prop-existence-violation:1}. holds.

\noindent

\end{fullproof}
\begin{sketch}
\begin{proof}[Sketch of proof.]
  For the left-to-right direction of \ref{item:prop-existence-violation:1}., we show that we can chase $I$ by using the sequence of dependencies $\sigma_0, \sigma_1, \sigma^{\tp}_1, \ldots, \sigma_n, \sigma^{\tp}_n, \sigma, \sigma'$ in this order, where $\pi = \sigma_0, \ldots, \sigma_n$, and the $\sigma^\tp_i$ are the corresponding $\tp$-dependencies as in Definition~\ref{def:accessible}.
  The instance resulting from this chase sequence includes the violation $\{T(a), \Triple(a,p,b), \Triple(a,p,b')\}$, and is included in any universal solution for $I$ to $\Sigma_\st \cup \Sigma_\Sbold^\tp$.

  For the right-to-left direction of \ref{item:prop-existence-violation:1}., the existence of $\pi$ follows from Lemma~\ref{lem:existence-type-fact}. 
  We then use a similar technique as in the proof of Lemma~\ref{lem:existence-type-fact}: we take an arbitrary terminating chase sequence $s$ of $I$ with $\Sigma_\st \cup \Sigma_\Sbold^\tp$, and we show that this sequence necessarily contains applications of dependencies $\sigma,\sigma'$ as required.
  Item~\ref{item:prop-existence-violation:2}. is proved simultaneously to item~\ref{item:prop-existence-violation:1}.
\end{proof}
\end{sketch}

\paragraph{Proof of Theorem~\ref{thm:pre-des-cone}}
\begin{fullproof}
We first show the left-to-right direction by proving its contraposition.
Let a violation sort $(T,f,p)$, and let $\pi,\sigma,\sigma',J,h$ be as in the theorem, in particular $h\circ\hpisigma(t) \neq h\circ\hpisigma(t')$.
Then by Proposition~\ref{prop:existence-violation} we have that $J_0$, the core pre-solution for $I$ to $\Sigma_\st \cup \Sigma_\Sbold^\tp$, includes $w = \{T(a), \Triple(a,p,b), \Triple(a,p,b')\}$, where $a = h\circ\hpisigma(f(\bar{z})) = h\circ\hpisigma(f(\bar{z}'))$, $b = h\circ\hpisigma(t)$ and $b' = h\circ\hpisigma(t')$.
By hypothesis, $b \neq b'$, so $w$ is a $(T,f,p)$-violation, so the core pre-solution of $I$ does not satisfy $\shexmostonedep$.

We show the right-to-left direction again proving its contraposition.

	Suppose there exists a consistent source instance $I$ s.t. its core pre-solution $J_0$ includes a $(T,f,p)$-violation, say $w = \{T(a), \Triple(a,p,b), \Triple(a,p,b')\}$ with $b \neq b'$ and $(T,f,p)$ a violation sort.
So Proposition~\ref{prop:existence-violation} applies allowing to deduce that there exist $\pi,\sigma,\sigma',h$ s.t. $(T,f)$ is accessible with $\pi$ in $\E$, $\sigma,\sigma'$ are $(T,f,p)$-contentious st-tgds, and $h: \Instpisigma \to I$ is a homomorphism, with $\head(\sigma) = \Triple(f(\bar{z}), p, t)$ and $\head(\sigma') = \Triple(f(\bar{z'}), p, t')$, and $a = h\circ\hpisigma(f(\bar{z})) = h\circ\hpisigma(f(\bar{z}'))$, $b = h\circ\hpisigma(t)$ and $b' = h\circ\hpisigma(t')$.
Therefore  $h\circ\hpisigma(t) \neq h\circ\hpisigma(t')$.
Suppose by contradiction that for all $J,h'$ s.t. $J$ solution for $\Instpisigma$ to $\Sigma_\fd$ and $h':\Instpisigma \to J$ the corresponding homomorphism we have $(\hpisigma \circ h')(t) = (\hpisigma \circ h')(t')$. We are then able to construct a contradiction to the fact that $I$ is a consistent source instance that satisfies the source functional dependencies.
\end{fullproof}

\begin{sketch}
  \paragraph{Sketch of proof of Theorem~\ref{thm:consistency}}
  The contraposition of the left-to-right direction easily follows for Proposition~\ref{prop:existence-violation}.
  We also show the contraposition of the right-to-left direction.
  Take a consistent source instance $I$ that does not have a solution to $\E$, then by Lemma~\ref{lem:only-tc-deps-matter} all solutions for $I$ to $\Sigma_\st \cup \Sigma_\Sbold^\tp$ contain some $(T,f,p)$-violation, and by Proposition~\ref{prop:existence-violation} we identify the corresponding $\pi,\sigma,\sigma',h$.
\mycomment{TODO: to be finished}
\end{sketch}

\subsubsection{Node kind consistency}

We start with a slight generalization of the construction of a universal simulation solution presented in Section~\ref{sec:query-answering} that will be used for computing $\Cotypes(J)$ and for the proofs of Lemma~\ref{lem:not-ctwo-implies-inconsistent} and Lemma~\ref{lem:cone-and-ctwo-implies-solution-exists}.
We recall the definitions from Section~\ref{sec:query-answering} are useful for this definition.

Fix a source instance $I$, and suppose that $J$ is a solution for $I$ to $\sttc \cup \shexmostonedep$.
The types associated to a node in a typed RDF graph $\types_G(n) = \{T \mid T(n)\in G\}$. 
The IRIs required by a set of types $X$ is the set $\Req(X) = \{p \mid \exists S,\mu, T \in X.\ p\dbl S^\mu\in \delta(T)\}$.
The frontier of $J$ is the following set
\[
\Fbb = 
\{
(n,p) \mid n\in\nodes(J),\ p\in\Req(\types_{J}(n)),\ \nexists m.\,\Triple(n,p,m)\in J
\}.
\]
For any set of types $X$ and IRI $p$, $\Delta(X,p)$ is the set of types that must hold at any node having a type from $X$ and reachable by a $p$-labeled edge:
\[
\Delta(X,p) = \{S \mid \text{$p\dbl{}S^\mu\in\delta(T)$ for some $T\in X$ and $\mu\in\{\MAYBE,\ONE,\PLUS,\MANY\}$}\}.
\]

Whether $J$ can be augmented to a solution that satisfies $\shexdep$ depends on the sets of types that co-occur in the frontier of $J$.
Define the set of subsets of $\Tcal \cup \{\Literal\}$: $N_0=\{\Delta(\types_J(n),p) \mid (n,p)\in\Fbb\}$.
Then let $N_{J}=\bigcup_{i=0}^\infty N_i$, where $N_i=\{\Delta(X,p) \mid X \in N_{i-1},\ p\in\Req(X)\}$ for any $i \geq 1$.
This process reaches a fixed point in a final number of steps.

\begin{lemma}
  \label{lem:NJ-is-cotypes}
  For any typed graph $J$, $N_J = \Cotypes(J)$.
\end{lemma}
As a corollary of Lemma~\ref{lem:NJ-is-cotypes} we get that $\Cotypes(J)$ can be effectively computed.

\paragraph{Proofs of Lemma~\ref{lem:not-ctwo-implies-inconsistent} and Lemma~\ref{lem:cone-and-ctwo-implies-solution-exists}}

\begin{lemma}
  \label{lem:lit-blank-inconsistency}
  For any instance $I$ of $\Rcal$ and any $J$ solution for $I$ to $\sttc \cup \shexmostonedep$, if $N_{J}$ contains a set $X$ with $\{\Literal, T\} \subseteq X$ for some type $T$ in $\Tcal$, then $I$ does not admit a solution to $\E$ that includes $J$.
\end{lemma}
\begin{proof}[Sketch of proof]
  We first show that if $X$ is a set in $N_{J}$, then any solution $G$ for $I$ to $\E$ that includes $J$ must contain a node $n_X$ with $X \subseteq \types_G(n_X)$.
  This is done by induction on the index $i$ s.t. $X \in N_i$.
  Next we show that if $\Literal$ is in $X$, then $n_X$ must be a literal, and if some type $T$ from $\Tcal$ is in $X$, then $n_X$ must be an IRI or a blank node.
  Then the lemma follows by contradiction.
\end{proof}

Remark that Lemma~\ref{lem:not-ctwo-implies-inconsistent} is an immediate consequence of Lemma~\ref{lem:NJ-is-cotypes} and Lemma~\ref{lem:lit-blank-inconsistency}.

Now, if $N_{J}$ does not contain any set $X$ in which $\Literal$ co-occurs with some type $T$ from $\Tcal$, then we can construct a solution for $I$ to $\E$ that includes $J$, as follows.
For any $X \in N_{J}$ s.t. $X \subseteq \Tcal$, let $n_X$ be a fresh blank node, i.e. $n_X \in \NullIri \setminus \dom(J)$.
For any $X \in N_{J}$ and $p \in \Req(X)$, let $n_{X,p}$ be a fresh null literal, i.e. $n_{X,p} \in \NullLit \setminus \dom(J)$.
Define the graph $G_\Sbold$ as follows.
\begin{align*}
  G_\Sbold = &\{\Triple(n,p,n_{X}) \mid (n,p) \in \Fbb \land X = \Delta(\types_J(n),p) \subseteq \Tcal \} \cup{} \\
               & \{\Triple(n_X,p,n_{X,p}) \mid (n,p) \in \Fbb \land \Delta(\types_J(n),p) = \{\Literal\} \} \cup{} \\
               & \{\Triple(n_X,p,n_{X'}) \mid X \in N_{J} \land p\in \Req(X) \land X' = \Delta(X,p) \subseteq \Tcal \} \cup{} \\
               & \{\Triple(n_X,p,n_{X,p}) \mid X \in N_{J} \land p\in \Req(X) \land \Delta(X,p) = \{\Literal\}\} \cup{} \\
               & \{T(n_X) \mid X\in N_{J} \land T \in X \}.
\end{align*}

\begin{lemma}
  \label{lem:GS-solution}
  For any source instance $I$ and any $J$ solution for $I$ to $\sttc \cup \shexmostonedep$, if $N_{J}$ does not contain a set $X$ with $\{\Literal, T\} \subseteq X$ for some type $T$ in $\Tcal$, then $J \cup G_\Sbold$ is a solution for $I$ to $\E$.
\end{lemma}
\begin{proof}[Sketch of proof.]
  $J \cup G_\Sbold$ satisfies $\Sigma_\st$ as $J$ does.
  It is easy to see by its definition that $J \cup G_\Sbold$ also satisfies the $\tp$ and $\pe$ dependencies in $\shexdep$. Regarding the $\pf$ dependencies in $\shexdep$:
  on the one hand, $J$ satisfies the $\pf$ dependencies by hypothesis.
  On the other hand, by construction, the triples added in $G_\Sbold$ are such that no node has more than one $p$-outgoing edge for any IRI $p$. 
  Therefore $J \cup G_\Sbold$ does not contain a trigger for a $\pf$ dependency.
\end{proof}
We point out that Lemma~\ref{lem:cone-and-ctwo-implies-solution-exists} is an immediate consequence of Lemma~\ref{lem:NJ-is-cotypes} and Lemma~\ref{lem:GS-solution}.

\paragraph{Proof of Theorem~\ref{thm:consistency}}
\label{app:previous-section}

Let $\E = (\Rbold, \Sbold, \Sigma_\st,\Fbold)$ and let $\Sbold = (\Tcal, \delta)$.
We first show that it is decidable whether $\E$ is node kind consistent.

Let $N_0$ be the set of subsets of $\Tcal \cup \{\Literal\}$ such that $X \in N_0$ iff there exists a function symbol $f \in \Fcal$ s.t. $X = \{T \mid (T,f) \text{ accessible in }\E\}$.
Then we define $\Cotypes(\E) = \bigcup^\infty_{i=0} N_i$, where $N_i =\{\Delta(X,p) \mid X \in N_{i-1},\ p\in\Req(X)\}$ for any $i \ge 1$.
Note that  $\Cotypes(\E)$ converges to a fix point in a finite number of steps.

\begin{lemma}
  \label{lem:toto}
  $\E$ is node kind consistency iff $\Cotypes(\E)$ does not contain a set $X$ s.t. $\{T, \Literal\} \subseteq X$ for some $T \in \Tcal$.
\end{lemma}
\begin{proof}
  We need to show that for any $I$ instance of $\Rcal$ $\Cotypes(J_0)$ does not contain a set $X$ with $\{T, \Literal\} \subseteq X$ for some $T \in \Tcal$ iff $\Cotypes(\E)$ does not contains a set $X$ s.t. $\{T, \Literal\} \subseteq X$ for some $T \in \Tcal$.
 
  For the left-to-right direction we show that if $\Cotypes(\E)$ contains such $X$, then there exists an instance $I$ s.t. $\Cotypes(J)$ contains a set $X'$ with $X \subseteq X'$. 
  It is enough to take the $I$ such that it contains exactly one fact for any relation in $\Rcal$, and such that $\dom(I) = \{b\}$ for some constant $b$.

  For the right-to-left direction, define the graph $R$ which vertices are $\Tcal$ and that has an edge labelled with $p$ from $T_1$ to $T_2$ iff $\delta(T_1)$ contains a triple constraint $p \dbl T_2^\ONE$ or $p \dbl T_2^\PLUS$.
  Intuitively, a $p$-edge from $T_1$ to $T_2$ in $R$ indicates that the label $p$ is required in every node that has type $T_1$, and every $p$-edge leads to a node that must have type $T_2$.

  We show that if there is an instance $I$ and a type $T \in \Tcal$ s.t. $\Cotypes(J_0)$ contains $X$ with $\{T, \Literal\} \subseteq X$ then necessarily the frontier of $J_0$ contains some $(n,p)$ s.t. $n$ is not null and there exist types $S$, resp. $S'$ in $\types_{J_0}(n)$ that are, intuitively, the reasons why $T$, resp. $\Literal$, were added to $X$ during the construction of $\Cotypes(J_0)$.
  Note that such $S,S'$ are not necessarily distinct.
  Moreover, there is a sequence $w$ of IRI's and a path in $R$ from $S$ to $T$ labelled with $w$ and a path in $R$ from $S'$ to $\Literal$ labelled with $w$.
  
  Using Lemma~\ref{lem:existence-type-fact} we deduce that $(S,f)$ and $(S',f)$ are accessible in $\E$, where $f$ is the function symbol s.t. $n \in \ran(f^F)$.
  Then we show inductively on the $w$ that during the construction of $\Cotypes(\E)$ we will reach a set $X'$ that contains both $T$ and $\Literal$.
\end{proof}

We now describe a coNP decision procedure for $\E$ being node kind consistent.
A certificate for a node kind inconsistency is composed of types $T,S,S'$ and a function symbol $f \in \Fcal$ as in the proof of the right-to-left direction of Lemma~\ref{lem:toto}.
More precisely, choose non-deterministically $T,S,S'$ and $f$ s.t. $(S,f)$ and $(S,f')$ are accessible in $\E$ (the latter can be tested in polynomial time).
According to the proof of Lemma~\ref{lem:toto}, it is enough to test whether there exists a sequence $w$ of IRIs s.t. $R$ has paths labelled with $w$ from $S$ to $T$ and from $S'$ to $\Literal$.
The latter can be polynomially tested by considering two finite state automata $A_S$ and $A_{S'}$ that are both derived from the graph $R$.
That is, both automata have the vertices of $R$ as states and the edges of $R$ as transitions.
$A_S$ has $S$ as initial state, while $S'$ is the initial state of $A_{S'}$.
We then compute the product automaton $A_S \times A_{S'}$ in polynomial time.
Then there exists $w$ as above iff state $(T,\Literal)$ is accessible in the product automaton.
In this case, $w$ is the shortest path from $(S,S')$ to $(T,\Literal)$ in this automaton.

Now the proof of Theorem~\ref{thm:consistency} can be completed:
\begin{itemize}
\item If $I$ is an instance of $\Rbold$ and the core pre-solution for $I$ to $\E$ is value consistent and node kind consistent, then $I$ admits a solution to $\E$ (Lemma~\ref{lem:cone-and-ctwo-implies-solution-exists}).
\item If $I$ is an instance of $\Rbold$ and the core pre-solution of $I$ to $\E$ is not value consistent, resp. is not node kind consistent, then $I$ does not admit a solution to $\E$ (corollary of Theorem~\ref{thm:pre-des-cone}, resp. Lemma~\ref{lem:not-ctwo-implies-inconsistent}).
\item It is decidable whether $\E$ is value consistent (Lemma~\ref{cor:des-cone}) and it is decidable whether $\E$ is node kind consistent (here above).
\end{itemize}

\subsubsection{Proof of Theorem~\ref{thm:consistency-intractable}}

\paragraph{Upper bound}

Checking node kind consistency is in co-NP as shown in Section~\ref{app:previous-section}.

Regarding value consistency:
\begin{lemma}
  \label{cor:des-cone}
  Deciding whether $\E$ is value consistent is in coNP.
\end{lemma}
\begin{proof}
    The contraposition of Theorem~\ref{thm:pre-des-cone} implies that $\E$ is value inconsistent iff there exists a source instance ($J$ in the theorem) that satisfies the source integrity constraints $\Sigma_\fd$ but is value inconsistent.
  This gives a co-NP decision procedure for value consistency.
  Indeed, $J$ as in the theorem is a certificate for the value inconsistency.
  We now argue that such certificate has size polynomial in the size of $\E$ and we can test in polynomial time whether it is indeed value inconsistent.
  First, guess a violation sort $(T,f,p)$, an elementary sequence $\pi = \sigma_0,\ldots,\sigma_n$ and two st-tgds $\sigma,\sigma'$ from $\Sigma_\st$.
  This is done in polynomial time as $\pi$ is elementary.
  Then check that $(T,f)$ is accessible in $\E$ with $\pi$ and that $\sigma,\sigma'$ are contentious with sort $(T,f,p)$ and construct the source instance $\Instpisigma$.
  This step is done in polynomial time as well.
  Finally, chase $\Instpisigma$ with $\Sigma_\fd$.
  The latter can also be done in polynomial time because the instance $\Instpisigma$ has a polynomial size, and all bodies of dependencies in $\Sigma_\fd$ contain exactly two atoms, thus require to compute a unique join in order to be evaluated.
  Additionally, $\Sigma_\fd$ chase steps do not increase the size of the instance, and only a polynomial number of chase steps can be executed before a solution or a failure is reached.
  The result of the chase is the certificate $J$.
  Consequently, deciding whether $\E$ is value consistent is in coNP.
\end{proof}

\paragraph{Lower bound}
  We prove coNP-hardness with reduction from the complement of SAT. Take any CNF

  $\varphi=c_1\land\ldots\land c_m$, where
  $c_j=\ell_{j,1}\lor\ldots\lor\ell_{j,k_j}$ is a clause over the variables
  $x_1,\ldots,x_n$. We construct the corresponding data exchange setting
  $\E_\varphi$ as follows. The relational schema consists of the following
  binary relation names (each having the first attribute as a key
  $A\rightarrow B$)
  \[
    V_\T(\underline{A},B), V_\F(\underline{A},B), R_1(\underline{A},B), \ldots,
    R_m(\underline{A},B)
  \]
  The constructor set is 
  \[
    \Fcal=\{f_1,\ldots,f_m,f_{m+1}\}
  \] 
  and their implementation is very straightforward
  $f_i(x) = \text{"$i$:"}+\mathtt{str}(x)$. We use the types
  \[
    \Tcal=\{T_1,\ldots,T_m,T_{m+1}\}
  \] 
  and the shape constraints:
  \begin{align}
    & T_j\to a\dbl T_{j+1}^\MANY\qquad\text{for $1\leq j \leq m$ and} \\
    & T_{m+1}\to a\dbl \Literal^\ONE.
  \intertext{
  The source to target dependencies are as follows. First, we have the two rules:
  }
    &V_\T(x,y) \Rightarrow \Triple(f_{m+1}(x),a,y) \label{eq:VT}\\
    &V_\F(x,y) \Rightarrow \Triple(f_{m+1}(x),a,y) \label{eq:VF}
  \intertext{
  Next, for any $1\leq j\leq m$ let $c_j=\ell_{j,1}\lor\ldots\lor\ell_{j,k_j}$
  and for $q\leq k \leq k_j$ if $\ell_{j,k}= x_i$, then we add this rule
  }
    &R_i(x,y)\land V_\T(x,y)\Rightarrow \Triple(f_j(x),a,f_{j+1}(x))\label{eq:RT}
  \intertext{
  and otherwise if $\ell_{j,k}= \lnot x_i$, then we add this rule
  }
    &R_i(x,y)\land V_\F(x,y)\Rightarrow \Triple(f_j(x),a,f_{j+1}(x))\label{eq:RF}
  \intertext{
  And finally, we add the following two rules:
  }
    &V_\T(x,y)\Rightarrow T_1(f_1(x))\\
    &V_\F(x,y)\Rightarrow T_1(f_1(x))
  \end{align}
  We claim that 
  \[
    \varphi\in\text{SAT}\quad\text{iff}\quad\text{$\E_\varphi$ is not consistent}.
  \]

  For \emph{only if} part, we take a valuation $V$ that satisfies $\varphi$ and
  construct an instance $I_V$ as follows. We fix 3 constants $c$, $\T$, and
  $\F$. The instance is
  \begin{multline*}
    I_V=\{V_\T(c,\T), V_\F(c,\F)\}\cup \{R_i(c,\T) \mid i\in\{1,\ldots,n\},\ V(x_i)=\text{\bf true}\}\cup{}\\
    \{R_i(c,\F) \mid i\in\{1,\ldots,n\},\ V(x_i)=\text{\bf false}\}.
  \end{multline*}
  It is easy to see that $I_V$ is consistent and with a simple inductive proof
  we can show that the result of chase on $I_V$ contains $T_{m+1}(f_{m+1}(c))$
  and the two triples $\Triple(f_{m+1}(c),a,\T)$ and $\Triple(f_{m+1}(c),a,\F)$
  which violates the shape constraint on the type $T_{m+1}$.

  For the \emph{if} part, we take a consistent instance $I$ such that chase of $I$ with $\E_\varphi$ is equal to $J$ and violates the shape constraints. The only shape
  constraint that can be violated is the constraint on the type $T_{m+1}$ (all
  remaining constraints can be satisfied by chase by adding null values if
  needed). Consequently $J$ contains $T_{m+1}(t_{m+1})$,
  $\Triple(t_{m+1},a,\T)$, and $\Triple(t_{m+1},a,\F)$, for some $t_{m+1}$,
  $\F$, and $\T$. Naturally, the two triples must be introduced with the rules
  \eqref{eq:VT} and \eqref{eq:VT}, and therefore, there is a constant $c$ such
  that $t_{m+1}=f_{m+1}(c)$, $V_\T(c,\T)\in I$, and $V_\F(c,\F)\in
  I$. Furthermore, with a simple inductive proof we can show that for every
  $j\in\{1,\ldots,m\}$ we have $T_j(f_j(c))\in J$,
  $\Triple(f_j(c),a,f_{j+1}(c))\in J$. We observe the triples
  $\Triple(f_j(c),a,f_{j+1}(c))$ can be only added by chase with the use of
  rules \eqref{eq:RT} and \eqref{eq:RF}, and the inductive proof also shows that
  every clause $c_j$ has at least one literal for which the corresponding rule
  must have been triggered. Since $I$ is consistent for no $i\in\{1,\ldots,n\}$
  can $I$ have both $R_i(c,\T)$ and $R_i(c,\F)$ ($I$ may have none of the
  two). We can therefore define the following valuation
  \[
    V(x_i) = 
    \begin{cases}
      \text{\bf true} &\text{if $R_i(c,\T)\in I$,}\\
      \text{\bf false} & \text{otherwise.}
    \end{cases}
  \]
  We show that $V$ satisfies $\varphi$ by observing that if for the chase
  triggers a clause \eqref{eq:RT} or \eqref{eq:RF} that corresponds to some
  literal $\ell$ of $c_j$, then $V$ satisfies $c_j$. We finish the proof by
  observing that the proposed reduction is polynomial. 

\subsection{Consistency of Non-Constructive st-tgds}
\label{app:consistency-non-constructive}

\newcommand{\coldash}{\mathrel{\text{:--}}}
\newcommand{\espace}{\hspace{0.7cm}}

Using the notations from Section~\ref{sec:consistency-non-constructive}, recall that any pair of rules in $\Sigma_\st$ use pairwise disjoint variables, $\theta$ is the set of terms that appear in the heads of $\Sigma_\st$, and $\Tcal$ is the set of type names of the shapes schema $\Sbold$. Denote $V$ the universally quantified variables that appear in $\Sigma_\st$.
We write $A \in \Sigma_\st$ when the atom $A$ appears in some head in $\Sigma_\st$, thus $\Sigma_\st$ is viewed as a monadic relation over atoms.
We use $t,t',u,u'$ to denote terms, and $x,x',y,y'$ and $\bar{x},\bar{y}$ to denote variables and vectors of variables, respectively.
The relations $\Acc \subseteq \theta \times \Tcal$, $\Eq \subseteq \theta \times \theta$ and $\Rev \subseteq \theta$ are defined by the following mutually recursive rules, where $X,X',Y,Y'$ are variables over $\theta$.
\begin{align}
  &T(t) \in \Sigma_\st \coldash \Acc(T,t)\\
  &\Eq(f(\bar{x}), f(\bar{y})) & \text{ whenever } f(\bar{x}) \in \theta, f(\bar{y}) \in \theta \\ 
  &\Rev(x) & \text{ whenever } x \in V \\ 
  &\Rev(f(\bar{x})) &\text{ whenever } f(\bar{x}) \in \theta\\
  &\Acc(T,X), \Triple(X,\p,Y) \in \Sigma_\st \coldash \Acc(U,Y) & \text{ for every } \delta(T,\p) = U^\mu\\
  &\Acc(T,X), \Eq(X,Y) \coldash \Acc(T,Y) & \text{ for every } T \in \Tcal\\
  &\Triple(X,\p,Y) \in \Sigma_\st, \Triple(X',\p,Y') \in \Sigma_\st, \nonumber\\
  &\espace \Eq(X,X'), \nonumber\\
  &\espace \Acc(T,X), \Acc(T,X') \coldash \Eq(Y,Y') & \text{ for every } \delta(T,\p) = U^\mu, \mu \in \{\ONE, \MAYBE\} \\
  &\Triple(X,\p,Y) \in \Sigma_\st, \Triple(X',\p,Y') \in \Sigma_\st, \nonumber\\
  &\espace \Eq(X,X'), \Rev(Y), \nonumber\\
  &\espace \Acc(T,X), \Acc(T,X')  \coldash \Rev(Y') &\text{ for every } \delta(T,\p) = U^\mu, \mu \in \{\ONE, \MAYBE\}
\end{align}
Two terms are \emph{contentious} if they satisfy the relation $\Cont \subseteq \theta \times \theta$ defined by
\begin{align}
  &\Triple(X,\p,Y) \in \Sigma_\st, \Triple(X',\p,Y') \in \Sigma_\st, \nonumber\\
  &\espace \Eq(X,X'), \Acc(T,X), \nonumber\\
  &\espace \Rev(Y), \Rev(Y') \coldash \Cont(Y,Y') &\text{ for every } \delta(T,\p) = U^\mu, \mu \in \{\ONE, \MAYBE\} 
\end{align}
This captures the fact that triples generated from the atoms $\Triple(X,\p,Y)$ and $\Triple(X',\p,Y')$ might lead to value inconsistency caused by the functional predicate egd for $\delta(T,\p) = U^\mu$.

The rules defining $\Acc$, $\Eq$, $\Rev$ and $\Cont$ can be turned into a Datalog program over the signature $\{\Acc, \Eq, \Rev, \Cont, \_ \in \Sigma_\st\}$ by creating as many rules as required by the conditions on $\delta$.
The size of $P$ is polynomial in the size of $\E$.
Then we can use $P$ to materialize in polynomial time the relations $\Acc$, $\Eq$, $\Rev$ and $\Cont$.

\begin{example}[Example~\ref{ex:illustartion-acc-eq-rev} continued.]
  \label{ex:ex1}
  With the data exchange setting as defined in Example~\ref{ex:illustartion-acc-eq-rev}, the materialized relations contain the following facts (non exhaustive).
  \begin{align*}
    &\Acc(T, f(x')), \Acc(U,g(y'')), \Acc(U, g(y))\\
    &\Rev(w''), \Rev(y')\\
    &\Eq(g(y''),z'), \Eq(g(y),z'), \Eq(g(y),g(y'')), \Eq(f(x),f(x'))
  \end{align*}
  \qed
\end{example}

Let $\Cont(t,t')$ for some terms $t,t'$ in $\theta$.
A \emph{proof tree} for $\Cont(t,t')$ is a derivation of the program $P$ which root is $\Cont(t,t')$ and which other nodes are facts from the $\Acc$, $\Eq$, $\Rev$ and $\_ \in \Sigma_\st$ relations.
In particular, the children of the root of such proof tree are $\Triple(u,\p,t) \in \Sigma_\st$, $\Triple(u',\p,t') \in \Sigma_\st$, $\Eq(u,u')$, $\Acc(T,u)$, $\Rev(t)$ and  $\Rev(t')$ for some terms $u,u'$ and some type $T$ and predicate $\p$.

\begin{example}[Example~\ref{ex:ex1} continued.]
  \label{ex:ex2}
  With the data exchange setting as defined in Example~\ref{ex:illustartion-acc-eq-rev}, there is a proof tree for $\Cont(w'',y')$ which nodes are exactly the facts listed in Example~\ref{ex:ex1} and has additionally as leaves all the facts of the form $A \in \Sigma_\st$ for all the atoms $A$ that appear in some rule head in $\Sigma_\st$:
  \begin{align*}
    &\Triple(f(x),\p,g(y)) \in \Sigma_\st, \Triple(g(y),\q,z) \in \Sigma_\st\\
    &T(f(x')) \in \Sigma_\st, \Triple(f(x'),\p,z') \in \Sigma_\st, \Triple(z', \q, y') \in \Sigma_\st\\
    &\Triple(g(y''),\q,w'') \in \Sigma_\st
  \end{align*}
  \qed
\end{example}

With every such proof tree $\pi$ we can associate an instance $I_{\pi,t,t'}$ of $\Rcal$ s.t. when chased with $\Sigma_\st \cup \Sigma_\Sbold^\tp$ would produce a violation in which (two constants derived from) the terms $t,t'$ need to be equated by an functional predicate egd for $\delta(T,\p)$.
The instance $I_{\pi,t,t'}$ is effectively constructed by a \emph{backchase} procedure.
We claim that
\begin{lemma}
  If $I$ instance of $\Rcal$ is value inconsistent, then there exist two terms $t,t'$ in $\Sigma_\st$ s.t. $\Cont(t,t')$ holds and there is $I' \subseteq I$ and $\pi$ a proof tree for $\Cont(t,t')$ s.t. $I'$ is isomorphic $I_{\pi,t,t'}$.
\end{lemma}
The proof is similar the proof of Theorem~\ref{thm:pre-des-cone}.

Then in order to check value inconsistency of $\E$ it is enough to enumerate all $t,t'$ s.t. $\Cont(t,t')$, all proof trees $\pi$ for $\Cont(t,t')$ and the corresponding instances $I_{\pi,t,t'}$.
If such $I_{\pi,t,t'}$ exists and is a valid instance of $\Rbold$, that is, satisfies the source functional dependencies, then $I_{\pi,t,t'}$ is value inconsistent for $\E$ thus $\E$ is value inconsistent.
If there is no $I_{\pi,t,t'}$ that satisfies the source functional dependencies, then $\E$ is value consistent.

\begin{example}[Example~\ref{ex:ex2} continued.]
  With the data exchange setting as defined in Example~\ref{ex:illustartion-acc-eq-rev} and the proof tree mentioned in Example~\ref{ex:ex2}, we construct the source instance
  $$
  \{R(x,y,z), S(x, y'), R(x'', y, z'')\}.
  $$
  One can see that when chasing the above source instance we can derive the facts
  \begin{align*}
    &\{\Triple(f(x),\p,g(y)), \Triple(g(y),\q,\bot_1), U(g(y))\\
    &\hspace{1cm}T(f(x)), \Triple(f(x), \p, \bot_2), \Triple(\bot_2,\q,y'), U(\bot_2)\\
    &\hspace{1cm}\Triple(g(y),\q,w'')\} 
  \end{align*}
  Then using the functional predicate egd for $\delta(T,\p) = U^\ONE$ we \emph{reveal} $\bot_2$ as being equal to $g(y)$.
  Finally applying the functional predicate egd for $\delta(U,\q) = \Literal^\MAYBE$ we need to equate $w''$ and $y'$ thus the chase fails and $I$ is value inconsistent.
  \qed
 \end{example}

On the other hand, node kind consistency of a non-constructive data exchange setting $\E$ can be tested in the same way as for constructive data exchange settings.
This concludes the proof of Theorem~\ref{thm:consistency-non-constructive}.

\subsection{Complexity of consistency for nondeterministic shape schemas}
\label{app:complexity-non-det}
\renewcommand{\t}{\ensuremath{\mathbf{t}}}
\newcommand{\f}{\ensuremath{\mathbf{f}}}

We reduce the problem of validity of $\pmb{\forall\exists}\mathsf{QBF}$ formulas
to testing the consistency of constructive data exchange settings with
nondeterministic shape schemas.

We fix a formula $\Phi=\forall \bar{x}.\exists \bar{y}.\varphi$, where
$\varphi=c_1\land\ldots\land c_k$ is a conjunction of clauses over
$\bar{x}=x_1,\ldots,x_n$ and $\bar{y}=y_1,\ldots,y_m$. We construct the
following data exchange setting $\E_\Phi$. The relational source schema consists
of relations (each relation with a single key)
\[
V_\t(\underline{x},y),
V_\f(\underline{x},y),
R_{x_1}(\underline{x},y),
\ldots,
R_{x_n}(\underline{x},y).
\]
We employ a non-overlapping library that for every variable
$v\in\bar{x}\cup\bar{y}$ contains the unary IRI constructors
$f_v, f_v^\t, f_v^\f$ and for every clause $c$ it contains a unary $f_c$.

The source-to-target dependencies and the shape schema will introduce a gadget
for every variable $v\in\bar{x}\cup\bar{y}$ that will be the only possible
source of inconsistency. We identify 3 forms of the gadget, $\mathbf{G}_v$ when
the valuation of $v$ is (yet) undetermined, and $\mathbf{G}_v^\t$ and
$\mathbf{G}_V^\f$ for when the variable takes the value true and false
respectively. The 3 kinds of gadgets are presented in
Figure~\ref{fig:gadget-variable}
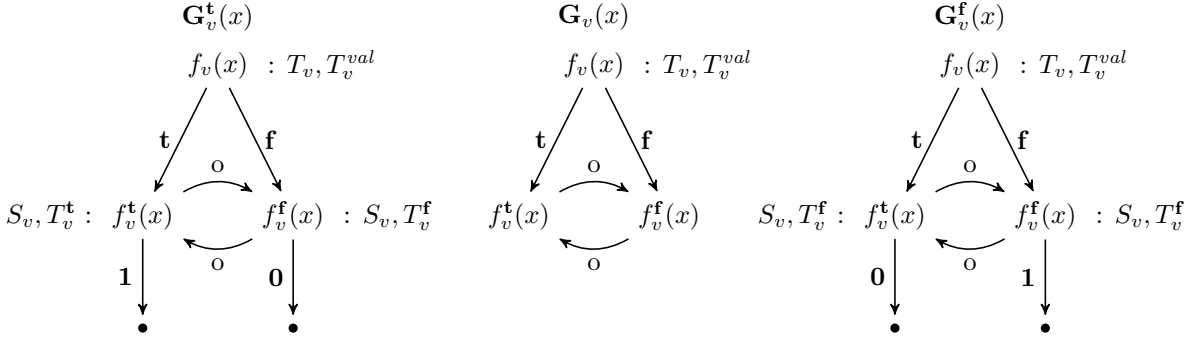
\begin{figure}[htb]
  \centering
  \begin{tikzpicture}[>=stealth',semithick,
      punkt/.style={circle,minimum size=0.1cm,draw,fill,inner sep=0pt, outer sep=0.125cm}]
    \path[use as bounding box] (-6,-3.5) rectangle (6,0.75);
    \begin{scope}
      \node at (0,0.65) {$\mathbf{G}_v(x)$};
      \node (v) at (0,0) {$f_v(x)$};
      \node[right] at (v.east) {: $T_v,T_v^\mathit{val}$};
      \node (vt) at (-1,-2) {$f_v^\t(x)$} edge [<-] node[left] {\t} (v);
      \node (vf) at (1,-2) {$f_v^\f(x)$} 
      edge [<-] node[right] {\f} (v)
      edge [<-,bend right] node[above] {o} (vt)
      edge [->,bend left] node[below] {o} (vt);
    \end{scope}
    \begin{scope}[xshift=-5cm]
      \node at (0,0.65) {$\mathbf{G}_v^\t(x)$};
      \node (v) at (0,0) {$f_v(x)$};
      \node[right] at (v.east) {: $T_v,T_v^\mathit{val}$};
      \node (vt) at (-1,-2) {$f_v^\t(x)$} edge [<-] node[left] {\t} (v);
      \node[left] at (vt.west) {$S_v,T_v^\t$ :};
      \node[punkt] (o) at (-1,-3.5) {} edge[<-] node[left]{$\mathbf{1}$} (vt);
      \node (vf) at (1,-2) {$f_v^\f(x)$}
      edge [<-] node[right] {\f} (v)
      edge [<-,bend right] node[above] {o} (vt)
      edge [->,bend left] node[below] {o} (vt);
      \node[right] at (vf.east) {: $S_v,T_v^\f$};
      \node[punkt] (o) at (1,-3.5) {} edge[<-] node[left]{$\mathbf{0}$} (vf);
    \end{scope}
    \begin{scope}[xshift=5cm]
      \node at (0,0.65) {$\mathbf{G}_v^\f(x)$};
      \node (v) at (0,0) {$f_v(x)$};
      \node[right] at (v.east) {: $T_v,T_v^\mathit{val}$};
      \node (vt) at (-1,-2) {$f_v^\t(x)$} edge [<-] node[left] {\t} (v);
      \node[left] at (vt.west) {$S_v,T_v^\f$ :};
      \node[punkt] (o) at (-1,-3.5) {} edge[<-] node[left]{$\mathbf{0}$} (vt);
      \node (vf) at (1,-2) {$f_v^\f(x)$}
      edge [<-] node[right] {\f} (v)
      edge [<-,bend right] node[above] {o} (vt)
      edge [->,bend left] node[below] {o} (vt);
      \node[right] at (vf.east) {: $S_v,T_v^\f$};
      \node[punkt] (o) at (1,-3.5) {} edge[<-] node[left]{$\mathbf{1}$} (vf);
    \end{scope}
  \end{tikzpicture}
  \caption{Three gadgets for a variable.}
  \label{fig:gadget-variable}
\end{figure}
The shape schema for every variable $v\in\bar{x}\cup\bar{y}$ contains the following types and their definitions 
\begin{align*}
  &T_v \rightarrow \t\dbl{} S_v,\ \f\dbl{}S_v\\
  &S_v \rightarrow o\dbl S_v, 
    \mathbf{1}\dbl{}T_\varnothing\MAYBE, 
    \mathbf{1}\dbl{}T_\varnothing\MAYBE\\
  &T_v^\mathit{val} \rightarrow 
  \t\dbl{}T_v^\t\MAYBE,\ 
  \f\dbl{}T_v^\t\MAYBE,\ 
  \t\dbl{}T_v^\f\MAYBE,\ 
  \f\dbl{}T_v^\f\MAYBE\\
  &T_v^\t \rightarrow 
    \mathbf{1}\dbl{}T_\varnothing\ONE,\ 
    \mathbf{0}\dbl{}T_\varnothing\NONE,\ 
    o\dbl{}T_v^\f\\
  &T_v^\f \rightarrow 
    \mathbf{1}\dbl{}T_\varnothing\NONE,\ 
    \mathbf{0}\dbl{}T_\varnothing\ONE,\ 
    o\dbl{}T_v^\t\\
  &T_\varnothing\rightarrow \epsilon
\end{align*}
We point out that imposing the types $T_v$ and $T_v^\mathit{val}$ on a node
$f_v(x)$ in $\mathbf{G}_v(x)$ creates a tension that can only be resolved by
adding the necessary outgoing edges $\mathbf{0}$ and $\mathbf{1}$ to the nodes
$f_v^\t(x)$ and $f_v^\f(x)$ but only in one of the two ways $\mathbf{G}_v^\t(x)$
or $\mathbf{G}_v^\f(x)$; in particular no node $f_v^c(x)$ can have both outgoing
edges $\mathbf{0}$ and $\mathbf{1}$ and every such node must have at least one
of the two. In fact, this is the principal reason why an result of chasing any
source instance would fail and we shall refer to such situation as type
overlap. The schema also contains the following type that shall be used to
enforce satisfiability of the clauses
\begin{align*}
&C   \rightarrow l\dbl{}T_\t\PLUS, l\dbl{}T_\f\MANY\\
&T_\t\rightarrow \mathbf{1}\dbl{}T_\varnothing\ONE\\ 
&T_\f\rightarrow \mathbf{0}\dbl{}T_\varnothing\ONE
\end{align*}
Now, the source-to-target dependencies are as follows. For every $x\in\bar{x}$
we have
\begin{align*}
  &V_\t(z,y)\land R_x(z,y) \Rightarrow G_x^\t(z)\\
  &V_\f(z,y)\land R_x(z,y) \Rightarrow G_x^\f(z)
\end{align*}
For every $y\in\bar{y}$ we have 
\begin{align*}
  &V_\t(x,y_1)\land V_\f(x,y_2) \Rightarrow G_y(x)
\end{align*}
For every clause $c$ we also introduce the following st-tgd. Let
$c=\ell_1\lor\ldots\lor\ell_a$, let $v_i$ be the variable used by the literal
$\ell_i$, and let $b_i\in\{\t,\f\}$ be the valuation of $v_i$ that satisfies
$c$. The st-tgd we introduce for $c$ is
\begin{align*}
  &V_\t(x,y_1)\land V_\f(x,y_2)\Rightarrow 
\Triple(f_c(x),l,f_{v_1}^{b_1}(x))\land\ldots\land
\Triple(f_c(x),l,f_{v_a}^{b_a}(x))\land C(f_c(x)).
\end{align*}
For instance, if $c=\lnot x_2\lor y_4 \lor x_3$, then the corresponding st-tgd
is
\begin{align*}
  V_\t(x,y_1)\land V_\f(x,y_2)\Rightarrow &
    \Triple(f_c(x),l,f_{x_2}^{\f}(x))\land{}\\
    &\Triple(f_c(x),l,f_{y_4}^{\t}(x))\land{}\\
    &\Triple(f_c(x),l,f_{x_3}^{\t}(x))\land{}\\
  &C(f_c(x)).
\end{align*}

\noindent
We claim that $\Phi$ is valid if and only if $\E_\Phi$ is consistent. 
\begin{itemize}
\item For the \emph{if} part we fix an instance $I$ and observe that a result of
  chasing $I$ may only fail due to type overlap in the nodes of some gadget
  $G_v(c)$. This is only possible if $c$ is present in both $V_\t(c,t)$ and
  $V_\f(c,f)$ for some $t$ and $f$. In fact, for each such $c$ we can consider
  the subset $I_c\subseteq I$ containing only the facts of $I$ that use $c$ as
  their key, and the problem can be treated independently for $I_c$. From this
  subinstance we construct the (possibly partial) valuation
  $V_C:\bar{x}\rightarrow\{\t,\f\}$ and we take any
  $V^*:\bar{y}\rightarrow\{\t,\f\}$ such that $V_c\cup V^*\models\varphi$. We
  use $V^*$ to construct a consistent solution $J_c$ to $I_c$. The solution to
  $I$ is obtained from taking the union of all $J_c$'s (and the result of
  chasing any elements of $I$ that do not belong to any $I_c$ but those cannot
  create any inconsistency). 
\item For the \emph{only if} part we take any valuation
  $V:\bar{x}\rightarrow\{\t,\f\}$ and build the instance 
  \[
    I_V=\{V_\t(c,\t),V_f(c,\f)\}\cup\{R_x(c,V(x))\mid x\in\bar{x}\}
  \]
  Since this is a consistent source instance, it has a solution $J$ for
  $\E_\Phi$. This solution contains either the gadget $G_y^\t(c)$ or $G_y^\f(c)$
  for every universally quantified variable $y\in\bar{y}$. We construct the
  corresponding valuation $V^*:\bar{y}\rightarrow\{\t,\f\}$. That
  $V\cup V^*\models \varphi$ follows from the fact that $J$ satisfies the
  constraints imposed by the type $C$
\end{itemize}

\subsection{Proofs for Section~\ref{sec:query-answering}(Certain Query Answering)}
\label{app:proofs-cqa}
Note that $\bisimulates$ is an equivalence relation on nodes of $G$, and we
denote by $[n]$ the equivalence class of node $n$ and by
$\nodes(G)/_\bisimulates$ the set of all equivalence classes. For each
equivalence class $C\in\nodes(G)/_\bisimulates$ we fix an arbitrarily chosen
representative node $\eta_C\in C$. Now, the \emph{bisimulation quotient} of $G$,
denoted by $G/_\bisimulates$ is the graph
$G/_\bisimulates=\{(\eta_{[n]},p,\eta_{[m]}) \mid (n,p,m)\in G\}$. The choice of
the representative does not matter because a non-null value is bisimilar only to
itself, and consequently, every non-singleton equivalence class in
$\nodes(G)/_\bisimulates$ contains null values only. The bisimulation quotient
of a typed graph $(G,\typing)$ is the typed graph $(G/_\bisimulates,\typing')$,
where $\typing'(\eta_C) = \bigcup\{\typing(n)\mid n\in C\}$ for any
$C\in\nodes(G)/_\bisimulates$.
Take a regular acyclic pattern $E$. We claim.
\begin{lemma}
	\label{lem:regular-expression}
	For any two nodes $n$ of $G$ and $m$ of $H$ such that $n \simulates m$, for any $n'$ of $G$, $(n,n')\in \llbracket E \rrbracket_G$ implies there is a $m'$ of $H$ such that $(m,m')\in \llbracket E \rrbracket_H$ and $n'\simulates m'$
\end{lemma}
\begin{proof}
	The proof is by induction on the structure of $E$. The base cases are $E$ with $|E|=1$ ($E\in \{\epsilon,p,\Box, \langle \ell \rangle \})$. We have the following cases:
	\begin{itemize}
		\item When $E=\epsilon$. Since $n\simulates m$, trivially $(m,m)\in \llbracket \epsilon \rrbracket_H$.
		\item When $E=p$. Assume $(n,n')\in \llbracket p \rrbracket_G$. By semantics of $\llbracket p \rrbracket_G$ and $n\simulates m$, there is a $m'$ such that $(m,p,m') \in H$. Since $(m,p,m') \in H$, then $(m,m')\in \llbracket p \rrbracket_H$ and $n'\simulates m'$
		\item When $E=\Box$. Assume $(n,n')\in \llbracket \Box \rrbracket_G$. By semantics of $\llbracket \Box \rrbracket_G$ and $n\simulates m$, there is a $m'$ such that $(m,p,m') \in H$. Since $(m,p,m') \in H$, then $(m,m')\in \llbracket p \rrbracket_H$ and $n'\simulates m'$.
		\item When $E=\langle\ell\rangle$. Assume $(n,n)\in \llbracket \langle\ell\rangle \rrbracket_G$. Since $n\simulates m$ and $m\in \nodes(G)$ and $\llbracket \langle\ell\rangle \rrbracket_G\not=\emptyset$, then $\ell=n=m$. Since $\ell=m$, then $(m,m)\in \llbracket \langle\ell\rangle\rrbracket_H$ and $n\simulates m$.
	\end{itemize}
	Now assume (IH) that for every expression $E$ with $|E|<i$, we have that for any two nodes $n$ of $G$ and $m$ of $H$ such that $n \simulates m$, for any $n'$ of $G$, if $(n,n')\in \llbracket E \rrbracket_G$ then there is a $m'$ such that $(m,m')\in \llbracket E \rrbracket_H$ and $n'\simulates m'$. Let $E$ be a regular expression with $|E|=i$. Assume $n \simulates m$. We distinguish the following cases:
	\begin{enumerate}
		\item $E=E_1 + E_2$. We have that $|E_1|<i$ and $|E_2|<i$. Assume $(n,n')\in \llbracket E \rrbracket_G$. We have to prove that there is a $m'$ such that $(m,m') \in \llbracket E \rrbracket_H$ and $n'\simulates m'$. Let $m'$ be in $H$. By definition, $(n,n')\in \llbracket E_1 \rrbracket_G \cup \llbracket E_2 \rrbracket_G$. By IH, there is $m_1$ of $H$ such that $(m,m_1)\in \llbracket E_1 \rrbracket_G$ and $n' \simulates m_1$. Let $m_1=m'$. By IH, there is $m_2$ of $H$ such that $(m,m_2)\in \llbracket E_2 \rrbracket_G$ and $n' \simulates m_2$. Let $m_2=m'$. By $(m,m_1)\in \llbracket E_1 \rrbracket_G$ and $(m,m_2)\in \llbracket E_2 \rrbracket_G$ and $m_1=m=m_2$, we have that $(m,m')\in \llbracket E_1 \rrbracket_H \cup \llbracket E_2 \rrbracket_H$. Hence, we conclude $(m,m')\in \llbracket E \rrbracket_H$ and $n'\simulates m'$.
		\item $E=E_1\cdot E_2$. We have that $|E_1|<i$ and $|E_2|<i$. Assume $(n,n')\in \llbracket E \rrbracket_G$. By semantics of $\cdot$, $(n,n')\in \llbracket E_1 \rrbracket_G \circ \llbracket E_2 \rrbracket_G$. By composition of binary relations, there is $n_2$ such that $(n,n_2)\in \llbracket E_1 \rrbracket_G$ and $(n_2,n')\in \llbracket E_2 \rrbracket_G$. Since $(n,n_2)\in  \llbracket E_1 \rrbracket_G$ and $n\simulates m$, then there is $m_2$ such that $(m,m_2)\in \llbracket E_1 \rrbracket_H$. By applying IH, we have that $n_2\simulates m_2$, and together with $(n_2,n')\in \llbracket E_2 \rrbracket_G$, we obtain that $(m_2,m') \in \llbracket E_2 \rrbracket_H$ for some $m'$ in $H$. We can use IH to conclude that $n'\simulates m'$. By $(m,m_2)\in \llbracket E_1 \rrbracket_H$ and $(m_2,m') \in \llbracket E_2 \rrbracket_H$, we have that $(m,m')\in  \llbracket E_1 \rrbracket_H \circ \llbracket E_2 \rrbracket_H$. By semantics of $\cdot$, $(m,m')\in  \llbracket E_1 \cdot E_2\rrbracket_H$. Thus, we obtain $(m,m') \in \llbracket E \rrbracket_H$ and $n'\simulates m'$.

		\item $E=E'^*$. We have that $|E'|=i-1$. Assume $(n,n')\in \llbracket E \rrbracket_G$. Since $i>1$ and by semantics of $*$, we have that $(n,n')\in \bigcup_{2\leq k}\llbracket E' \rrbracket_G^k$. Applying transitive closure of binary relation, we obtain $(n,n')\in \llbracket E'\rrbracket_G^2 \cup \llbracket E''^{*}\rrbracket_G$ and $|E''^{*}|<i$. Then, we have that $(n,n')\in \llbracket E'\rrbracket_G \circ \llbracket E'\rrbracket_G$, or $(n,n')\in \llbracket E''^{*}\rrbracket_G$. By $n\simulates m$ and the prove of $E_1\cdot E_2$, we know that there is $(m,m')\in  \llbracket E'\rrbracket_G \circ \llbracket E'\rrbracket_H$ and $n'\simulates m'$. Applying IH in $|E''^{*}|<i$, there is $(m,m')\in \llbracket E''^{*}\rrbracket_H$ and $n'\simulates m'$. By the two statements above, we have $(m,m')\in \bigcup_{2\leq k}\llbracket E' \rrbracket_H^k$ and $n'\simulates m'$. Thus, we conclude $(m,m')\in \llbracket E \rrbracket_H$ and $n'\simulates m'$.
		\item $E=[E']$. Assume $(n,n)\in \llbracket [E] \rrbracket_G$ and $n\simulates m$. By definition of $[E']$, there is $n'$ such that $(n,n')\in  \llbracket E' \rrbracket_G$. Applying IH, there is $m'$ such that $(m,m')\in \llbracket E' \rrbracket_H$. By definition, $(m,m)\in\llbracket [E'] \rrbracket_H$ and $n\simulates m$. Thus, we conclude that $(m,m)\in \llbracket E \rrbracket_H$.
	\end{enumerate}
\end{proof}
Finally, we claim.
\begin{lemma}
	\label{lem:graph-regular}
	For any two graphs $G$ and $H$, if $G\simulates H$ then if $G\models E$ implies $H\models E$.
\end{lemma}
\begin{proof}
Take any two graphs $G$ and $H$. Assume $G\simulates H$ and $G\models E$. By definition of $G\models E$, $\llbracket E \rrbracket_G \not=\emptyset$, and in consequence there is a pair $(n,n')\in \llbracket E \rrbracket_G$. By $G\simulates H$, there is a node $m$ in $H$ such that $n\simulates m$. By lemma~\ref{lem:regular-expression} and $n\simulates m$ and $(n,n')\in \llbracket E \rrbracket_G$, there is $m'$ of $H$ such that $(m,m') \in \llbracket E \rrbracket_H$. Since $\llbracket E \rrbracket_H \not= \emptyset$, then $H\models E$.
\end{proof}

Take two graphs $G$ and $H$. Assume $G\models E$. By lemma~\ref{lem:graph-regular}, it holds that $H\models E$. Consequently, $E$ is robust under simulation.

\subsubsection{Proof of Theorem~\ref{thm:universal-satisfies-tree}}
Take any consistent instance $I$ of $\Rbold$. Now, we prove that the typed graph $\U = J_0 \cup G_\Sbold$, is a universal simulation solution. First, we define \emph{reachability} from a node $n$ with a path $\pi$ as follows:
\begin{align*}
	&\begin{aligned}
	\pazocal{R}_G(N,p)&{}= \{ n' \mid \exists n\in N.\,\Triple(n,p,n') \in G \}\\
	\pazocal{R}_G^*(n,\pi\cdot p)&{}=\pazocal{R}_G(\pazocal{R}_G^*(n,\pi),p)\\
	\pazocal{R}_G^*(n,\epsilon)&{}=\{n\}
	\end{aligned}
\end{align*}
Then, we extend the canonical function $\Delta$ as follows.
\begin{itemize}
\item $\Delta^*(X,\pi\cdot p)=\Delta(\Delta^*(X,\pi),p)$
\item $\Delta^*(X,\epsilon)=X$
\end{itemize}
Now, we claim \begin{lemma}
	\label{lem:same-path-in-uni-sol}
	For any solution $J$ to $\E$ for $I$ and for any frontier $(n_0,p_0)\in \Fbb$ and for any path in $\U$ of the form $\pi=p_0\cdot p_1 \cdot \ldots \cdot p_k$. Let $X=\Delta^*(\types_{J_0}(n_0),\pi)$. Then 
	\begin{enumerate}
		\item $\pi$ is also in $J$
		\item if $n=\pazocal{R}_\U^*(n_0,\pi)$ then $\types_\U(n)=X$
		\item $\forall m \in \pazocal{R}_J^*(n_0,\pi) .\, X\subseteq \types_J(m)$
	\end{enumerate}
\end{lemma}
\begin{proof}
	Take any solution $J$ for $I$ to $\E$ and any frontier $(n_0,p_0)\in \Fbb$. We prove by induction in the size of the path $\pi$. Let $|\pi|\leq k$. The base case is when $\pi$ is of size 1, i.e. $\pi=p_0$. Let $X=\Delta^*(\types(n_0),\pi)$. Then
		\begin{enumerate}
			\item For case 1. We know  $(n_0,p_0)\in \Fbb$ and because $p_0$ is in $\U$ then $p_0\in \Req(\types_{J_0}(n_0))$. Since $J$ is a solution then the IRIs required by $\types_{J_0}(n_0)$ must be satisfied. Since $p_0\in \Req(\types_{J_0}(n_0))$ then there is $m$ such that $\Triple(n_0,p_0,m) \in J$. Thus, $\pi$ is in $J$.
			\item For case 2. Assume $n=\pazocal{R}_\U^*(n_0,p_0)$. By definition of $\Delta^*$, we have $X=\Delta(\types_{J_0}(n_0),p_0)$. Then, we take a type $T \in \types_\U(n)$, and by definition, we have $T(n)\in G_\Sbold$, $n\in N$ and $T\in n$. By construction of $G_\Sbold$ and $(n_0,p_0)\in \Fbb$, we obtain $n\in N_0$. Also by construction of $N_0$ and $X$ definition, we obtain $n=\Delta(\types_{J_0}(n_0),p_0)=X$. Finally, by construction of $G_\Sbold$, we have that $T\in X.$ Similar process is done in left direction. Thus, we conclude that $\types_\U(n)=X$.
			\item For case 3. Take any $m\in \pazocal{R}_J(n_0,p_0)$, i.e. $\Triple(n_0,p_0,m)\in J$. By definition of $\Delta^*$, we have $X=\Delta(\types_{J_0}(n_0),p_0)$. Then, take any $T\in X$. Since $(n_0,p_0)\in \Fbb$ and $J$ is a solution, then there is a type $T'$ such that $T'(m) \in J$. Let $T'=T$. By $T(m) \in J$, we have that $T\in \types_J(m)$. Thus, we conclude that $X\subseteq \types_J(m)$.
	\end{enumerate}
Now we fix $k>1$ and assume (IH) for any path $\pi$ in $\U$ such that $|\pi|\leq k$ holds that (a) $\pi$ is also valid in $J$, and (b) if $n=\pazocal{R}_\U(n_0,\pi)$ then $\types_\U(n)=X$; and (c) for any $m \in \pazocal{R}_J(n_0,\pi)$ holds that $X\subseteq \types_J(m)$.

Let $|\pi|\leq k+1$. Take any path $\pi$ in $\U$ such that $|\pi|\leq k+1$. Let $X=\Delta^*(\types(n_0),\pi)$ and $\pi=\pi'\cdot p$ such that $|\pi'|\leq k$. We have the following cases:
		\begin{enumerate}			
			
			\item Case 1. By definition of path, there is $n_0,\ldots,n_{k+1}$ such that $(n_{i-1},p_i,n_i)\in \U$ for $i\in\{1\ldots,k,k+1 \}$ and $\pi'=p_1\cdot\ldots\cdot p_k$. Applying IH, we have that $\pi'$ is a path in $J$ and there are $m,m_k$ in $J$ such that  $\Triple(m,p_k,m_k)\in J$. Since $(n_0,p_0)\in \Fbb$ and $p_{k+1}\in\Req(X)$ and $J$ is a solution, then it holds that $\Triple(m_k,p_{k+1},m_{k+1})\in J$. Thus, $\pi$  is a path in $J$.
			\item Case 2. Assume $n=\pazocal{R}_\U^*(n_0,\pi'\cdot p)$. Let $X'=\Delta^*(\types_{J_0}(n_0))$. By definition of $\Delta^*$, it is equivalent to $X=\Delta(X',p)$. Let $n_1=\pazocal{R}_\U^*(n_0,\pi')$. By definition of $\pazocal{R}_\U^*$, we have $\pazocal{R}_\U(n_1,p)=n$. Applying IH, we obtain that $\types_\U(n_1)=X'$, and by definition of path $\pi'\cdot p$, we have that $\Triple(n_1,p,n)\in \U$. Now, we take any $T\in X$. By the statements above and considering that $p\in \Req(X')$ and by construction of $\U$, we have that $T(n)\in \U$, i.e. $T \in \types_\U(n)$. A similar process is done for proving the right direction. Thus, we conclude that $\types_\U(n)=X$.
			\item Case 3. Take $m'\in \pazocal{R}_J^*(n_0,\pi'\cdot p)$. By definition of $\pazocal{R}_J$, there is $\Triple(m,p,m')\in J$ where $m\in \pazocal{R}_J^*(n_0,\pi')$. Now, we take any $T\in X$ and applying the IH, we obtain that there is a type $T'$ such that $T'(m)\in J$. Since $J$ is a solution and $p \in \Req(X)$ and statements above, we have that $T(m')\in J$, i.e. $T\in \types_J(m')$. Thus, we conclude that $X\subseteq \types_J(m)$. 
	\end{enumerate}
\end{proof}
Next, we claim the following.
\begin{lemma}
	\label{lem:universal-simulated-by-solution}
	$\U$ is simulated by every solution $J$ for $I$ to  $\E$.
\end{lemma}
\begin{proof}
	We construct 
	\begin{multline*}
	R=\{(n,m)\in {J_0}\times {J_0}\mid n=m \}\cup{}\\\{(n,m)\mid \exists (n_0,p_0)\in \Fbb .\, \exists\pi=p_0\cdot p_1\cdot \ldots \cdot p_k .\ n \in \pazocal{R}_\U(n_0,\pi) \land m\in \pazocal{R}_J(n_0,\pi) \}.
	\end{multline*}
	We show that $R$ is a simulation of $\U$ by $J$.
	Then, we take any pair $(n,m)\in R$ and $p\in \Iri$. We have the following cases:(a) $n \in {J_0} \land (n,p)\not\in \Fbb$ and (b) $(n,p)\in \Fbb \lor n \in \nodes(G_\Sbold)$.
	
	For case a. We know that $(n,m) \in {J_0}\times {J_0}$ and $n=m$. We take $n'\in \nodes({J_0})$ such that $\Triple(n,p,n')\in {J_0}$. As a result of consider  $m'=n'$, we obtain $m'\in \nodes({J_0})$ then $m'\in \nodes(J)$. Since $\Triple(n,p,n') \in {J_0}$, we have $\Triple(m,p,m') \in {J_0}$, and by $(n',m') \in {J_0}\times {J_0}$, we conclude $(n',m')\in R$.
	
	For case b. We prove only when $n \in \nodes(G_\Sbold)$ since the other is implied by this proof. By $n\in \nodes(G_\Sbold)$ and $(n,m)\in R$, we have that $n=\pazocal{R}_\U^*(n_0,\pi)$ where $\pi=p_1\cdot \ldots \cdot p_k$ and $(n_0,p_1)\in \Fbb$ and $m\in \pazocal{R}_J(n_0,\pi)$. Then, we take $p,n'$ such that $\Triple(n,p,n')\in \U$, i.e. $n'\in \pazocal{R}_\U^*(n_0,\pi\cdot p)$ and $\pi\cdot p$ is valid in $\U$. By lemma~\ref{lem:same-path-in-uni-sol}, we have $\pi\cdot p$ is a path in $J$, i.e. there is a node $m'\in \pazocal{R}_J(n_0,\pi\cdot p)$. Thus, we conclude that $(n',m')\in R$.
\end{proof}
By lemma~\ref{lem:universal-simulated-by-solution}, $\U$ is a universal simulation solution. Then the proof is relatively straightforward, and for the \emph{if} part, it suffices
to use Lemma~\ref{lem:regular-is-rub} and for the \emph{only if} part, it
suffices to notice that a universal simulation solution is also a solution. 
\subsubsection{Proof of Theorem~\ref{thm:quotient-solution-exp}}
Before proving this theorem, we show that $\U_0$ is indeed the minimal universal simulation solution. We take
any universal simulation solution $\U$ and create an injective mapping from the
nodes of $\U_0$ to the nodes of $\U$. The mapping is an identity on $J_0$ which
is contained in any solution. Now, for a node $n$ of $G_\Sbold/\bisimulates$ we
observe that there must be at least one path $\pi$ from a frontier node $n_0$ to
$n$, and because $\U$ is simulated in $\U_0$, there exists at least one node $m$
in $\U$ that is reachable from $n_0$ by path $\pi$. Consequently, we map $n$ to
an arbitrary such $m$. Now, suppose that two different nodes $n_1$ and $n_2$ of
$G_\Sbold/\bisimulates$ are mapped to the same node $m$. Because $\U_0$ is a
bisimulation quotient and the nodes $n_1$ and $n_2$ are different, they are not
bisimilar. However, since $\U_0$ is simulated by $\U$, and vice versa, and $n_1$
is reachable with the same path in $\U_0$ as $m$ in $\U$ and $n_1$ is reachable
with the same path in $\U_0$ as $m$ in $\U$, $n_1$ is bisimilar to $m$ and $m$
is bisimilar to $n_2$. By transitivity of bisimulation, we get that
$n_1\bisimulates n_2$, a contradiction.

Take an instance $I$ of $\Rbold$. We construct a typed graph as follows $\U_0={J_0}\cup G_\Sbold/_\bisimulates$ where $G_\Sbold/_\bisimulates$ is the bisimulation quotient of $G_\Sbold$. We claim that
\begin{lemma}
	\label{lem:exist-universal}
	$\U_0$ is a universal simulation solution.
\end{lemma}
\begin{proof}
	The proof is similar to Lemma~\ref{lem:universal-simulated-by-solution}. We construct a relation as follows:
	\begin{align*}
	R=\{(n,m)\in {J_0}\times {J_0}\mid n=m \}\cup\{(C,m)\mid \exists (n_0,p_0)\in \Fbb .\, \exists\pi=p_0\cdot p_1\cdot \ldots \cdot p_k .\,\\ \exists n\in C. \,n \in \pazocal{R}_\U(n_0,\pi) \land m\in \pazocal{R}_J(n_0,\pi) \}.
	\end{align*}
	The first case is the same as lemma~\ref{lem:universal-simulated-by-solution}. The second case is proven when $C \in \nodes(G_\Sbold)/_\bisimulates$. Since $C\in \nodes(G_\Sbold)/_\bisimulates$ and $(C,m)\in R$, then there is $n\in C$ such that $n=\pazocal{R}_\U^*(n_0,\pi)$ where $\pi=p_1\cdot \ldots \cdot p_k$ and $(n_0,p_1)\in \Fbb$ and $m\in \pazocal{R}_J(n_0,\pi)$. Now, we take $p\in \Iri ,C'\in \nodes(G_\Sbold)/_\bisimulates$ such that $(\eta_C,p,\eta_C')\in \U_0$, i.e., there are $n\in C$ and $n'\in C'$ such that $(n,p,n')\in G_\Sbold$. From this fact, we have that $n'\in \pazocal{R}_{\U_0}^*(n_0,\pi\cdot p)$ and $\pi\cdot p$ is valid in $\U_0$. By lemma~\ref{lem:same-path-in-uni-sol}, $\pi\cdot p$ is valid in $J$, i.e., there is a node $m'\in \pazocal{R}_J(n_0,\pi\cdot p)$. As a consequence, we conclude that $(C',m')\in R$ yielding that $\U_0$ is a universal simulation solution.
\end{proof}
Next, we claim.
\begin{lemma}
	\label{lem:universak-great-size}
	For any universal simulation solution $\U$ it holds, $|\U|\geq|\U_0|$ 
\end{lemma}
\begin{proof}
We take any universal simulation solution $\U$ and create an injective
mapping from the nodes of $\U_0$ to the nodes of $\U$. The
mapping is an identity on $J_0$ which is contained in any universal simulation
solution. Now for a node $n$ of $G_\Sbold/\bisimulates$ we observe that there
must be at least one path $\pi$ from a frontier node $n_0$ to $n$, and because
$\U$ is simulated in $\U_0$, there exists at least one node $m$
in $\U$ that is reachable from $n_0$ by path $\pi$. Consequently, we map
$n$ to an arbitrary such $m$. Now, suppose that two different nodes $n_1$ and
$n_2$ of $G_\Sbold/\bisimulates$ are mapped to the same node $m$. Because
$\U_0$ is a bisimulation quotient and the nodes $n_1$ and $n_2$ are
different, they are not bisimilar. However, since $\U_0$ is simulated by
$\U$, and vice versa, and $n_1$ is reachable is reachable with the same
path in $\U_0$ as $m$ in $\U$ and $n_1$ is reachable with the
same path in $\U_0$ as $m$ in $\U$, $n_1$ is bisimilar to $m$
and $m$ is bisimilar to $n_2$. By transitivity of bisimulation, we get that
$n_1\bisimulates n_2$, a contradiction.	
\end{proof}
Finally, we claim.
\begin{lemma}
	\label{lem:size-universal-simulation}
	There is a polynomial formula such that for any $n,m\in \mathbb{N}$, there exists a data exchange setting $\E$ and instance $I$ of $\Rbold$ such that the size of $\U_0$ is asymptotic to $\exp(m)$ and it holds $|\E|+|I|\leq\mathit{poly}(n)$ where $\mathit{poly}(n)$ is the polynomial formula.
\end{lemma}
\begin{proof}
	Let $I=\{ R(1)\}$ and $\Sigma_\st$ contains only $R(x)\Rightarrow T(f(x))$ and $\Sbold$ be as in Figure~\ref{fig:shex-graph-size} that contain cycles of length $2,3,5,\ldots,$ prime numbers with one shape type name different such as $T_{23}, T_{34},$ and $T_{56}$. Let $P_m$ stands for the $m$-th prime number for $m\in \mathbb{N}$. When constructing the universal simulation solution $\U_0$ we can observe that  $|\U_0|\equiv 1(\mod 2)$ and $|\U_0|\equiv 1 (\mod 3)$ and so on. Then, we can apply the chinese reminder theorem such that $|\U_0|\equiv 1 (\mod k)$ such that $k=2*3\ldots*P_m$. The product of $m$ prime numbers is approximately $2*3*\ldots*P_m\leq 2^{2m}$.
		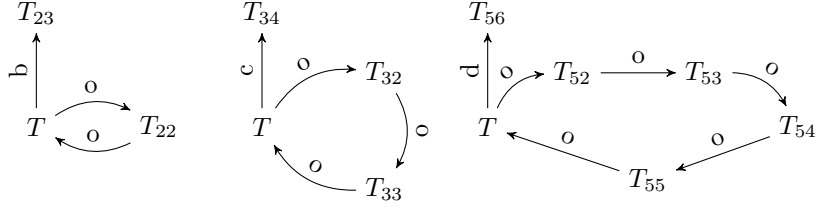
\begin{figure}
			\begin{tikzpicture}[->,>=stealth']
			\node (21) at (0,0) {$T$};
			\node[right=of 21](22)  {$T_{22}$};
			\node[above= of 21] (23){$T_{23}$};
			\node (31) at (3,0) {$T$};
			\node [above right=2mm and 10mm of 31](32) {$T_{32}$};
			\node [below=of 32](33) {$T_{33}$};
			\node[above= of 31] (34){$T_{34}$};
			\node (51) at (6,0) {$T$};
			\node [above right=2mm and 5mm of 51](52) {$T_{52}$};
			\node [right=of 52](53) {$T_{53}$};
			\node [right=3.5cm of 51](54) {$T_{54}$};
			\node [below right=2mm and 1.5cm of 51](55) {$T_{55}$};
			\node[above= of 51] (56){$T_{56}$};
			\path (21) edge[bend left] node[above,sloped]{o}  (22)
			(22) edge[bend left] node[above,sloped]{o}  (21);
			\path (31) edge[bend left] node[above,sloped]{o}  (32)
			(32) edge[bend left] node[above,sloped]{o}  (33)
			(33) edge[bend left] node[above,sloped]{o}  (31);
			\path (51) edge[bend left] node[above,sloped]{o}  (52)
			(52) edge[ left] node[above,sloped]{o}  (53)
			(53) edge[bend left] node[above,sloped]{o}  (54)
			(54) edge[ left] node[above,sloped]{o}  (55)
			(55) edge[ left] node[above,sloped]{o}  (51);
			\path (21) edge[left] node[above, sloped]{b} (23);
			\path (31) edge[left] node[above, sloped]{c} (34);
			\path (51) edge[left] node[above, sloped]{d} (56);
			\end{tikzpicture}
			\caption{Shape Schema Graph}
			\label{fig:shex-graph-size}
		\end{figure}
	We compute the size of the universal simulation solution $\U$ using the prime number counting function, denoted by $\pi(m)$ that  counts the number of primes less or equal to $m\in \mathbb{N}$. It follows from the prime number theorem that for all $m\in \mathbb{N},\pi(m) \sim m/\log(m)$, i.e.,$\lim\limits_{m \to \infty} (\pi(m)\times\log(m)/m)=1$. The prime number theorem guarantees that the set of all natural numbers up to a fixed size asymptotically contains an exponential number of prime number. By the prime number theorem, $P_m$ is asymptotic to $m*\log m$ as $m\to \infty$. The sum of $m$ prime numbers is 
	\begin{align*}
	2+3+\ldots+P_m &\leq m*P_m\\
	&\leq m*m*\log m\\
	&\leq m^3
	\end{align*}
	Let $n=|I|$. Since the application of $\Sigma_\st$ is founding an homomorphism in every tuple of $I$, in the worst case we can have that the size of the core pre-solution $|J_0|\leq n^2$. Let $\mathit{poly}(n)=n^2+1$. Finally, we get for $m,n\in \mathbb{N}$, the size of $|\U_0|\leq n^2+2^{2m/3}$. Let $\exp(m)=2^{2m/3}$. Thus, $\U_0$ is asymptotic to $\exp(m)$ and $|\E|+|I|\leq\mathit{poly}(n)$.	
\end{proof}
By lemma~\ref{lem:size-universal-simulation}, the size of $\U_0$ is bounded by a
polynomial in the size of $I$ and an exponential function in the size of
$\Sbold$. Since $\U_0$ exists, then we can construct a size-minimal universal simulation solution.

\subsubsection{Proof of Theorem~\ref{thm:complexity}}

Let $\E$ be a constructive data exchange setting with fixed IRI constructors, and $Q$ a regular acyclic pattern. Take any instance $I$ of $\Rbold$. 

Following the semantics of nSPARQL \cite{PEREZ2010255}, we find equivalences to NRE$^\rightarrow$ as follows:
\begin{align*}
&
\begin{aligned}
\llbracket \epsilon \rrbracket_G &{}=\llbracket \mathsf{self} \rrbracket_G ,
&
\llbracket [E] \rrbracket _G&{}= \llbracket \mathsf{self}\dbl[exp] \rrbracket _G,
\\
\llbracket p \rrbracket_G&{}=\llbracket \mathsf{next}\dbl a \rrbracket_G,
&
\llbracket E_1+E_2 \rrbracket_G &{}= \llbracket exp_1|exp_2 \rrbracket_G,
\\
\llbracket \Box \rrbracket_G&{}=\llbracket \mathsf{next} \rrbracket_G,
&
\llbracket E_1\cdot E_2 \rrbracket _G& {}=\llbracket exp_1/ exp_2 \rrbracket _G,
\\
\llbracket \langle \ell \rangle \rrbracket _G&{}=\llbracket \mathsf{self}\dbl a \rrbracket _G,
&
\llbracket E^* \rrbracket _G&=\llbracket exp^* \rrbracket _G,
\\
\end{aligned}
\end{align*}
where $\mathsf{next},\mathsf{self}$ are navigational axes that nSPARQL uses and $exp$ is an expression in nSPARQL where the axis of expression $\in \{\mathsf{next},\mathsf{self}\} $.

We use the polynomial decision algorithm presented in section 3.1 of \cite{PEREZ2010255} for evaluating $Q$ in $I$ w.r.t. $\E$. 
It is known in \cite{PEREZ2010255} that the evaluation of a query in a graph is $O(|G|\cdot |exp|)$. By theorem~\ref{thm:quotient-solution-exp}, we construct a size-minimal universal simulation solution for $I$ to $\E$ such that its size bounded by a polynomial in the size of $I$ and an exponential function in the size of $\Sbold$. Since $\Sbold$ is fixed and because the universal simulation solution is bounded, the data complexity is $O(|\U_0|\cdot |E|)$, which is $O((n^2+2^{c*m})\cdot |E|)$ where $n$ is the size of $I$ and $m$ is the number of shape types.


\subsection{Proof of Proposition~\ref{prop:full-nre-intractable}}
The proof is by reduction of intersection non-emptiness of $n$ regular
expressions $E_1,\ldots,E_n$ over $\Sigma$. Indeed, we only need a simple schema
$\delta(T,a)=T^\PLUS$ for $a\in\Sigma$, a single st-tgd
$R(x)\Rightarrow T(f(x))$, and a instance $I=\{R(0)\}$. If we let $\#=f(0)$,
then \emph{true} is the consistent answer to
$Q=\#\cdot \Box^* \cdot [E_1^-\cdot\#^-]\cdot \ldots\cdot [E_n^-\cdot\#^-]$ if
and only if $E_1\cap\ldots\cap E_n$ is nonempty.




\end{document}